%% file: dh_dmax.tex
\documentclass[11pt,notitlepage,tightenlines,nofootinbib,superscriptaddress,aps,pra]{revtex4-2}
\pdfoutput=1

\input{macros}

\begin{document}
\count\footins = 1000

\title{Tight relations and equivalences between smooth relative entropies}

\author{Bartosz Regula}
\email{bartosz.regula@gmail.com}
\affiliation{Mathematical Quantum Information RIKEN Hakubi Research Team, RIKEN Pioneering Research Institute (PRI) and RIKEN Center for Quantum Computing (RQC), Wako, Saitama 351-0198, Japan}

\author{Ludovico Lami}
\email{ludovico.lami@gmail.com}
\affiliation{Scuola Normale Superiore, Piazza dei Cavalieri 7, 56126 Pisa, Italy}
\affiliation{QuSoft, Science Park 123, 1098 XG Amsterdam, the Netherlands}
\affiliation{Korteweg--de Vries Institute for Mathematics, University of Amsterdam, Science Park 105-107, 1098 XG Amsterdam, the Netherlands}
\affiliation{Institute for Theoretical Physics, University of Amsterdam, Science Park 904, 1098 XH Amsterdam, the Netherlands}

\author{Nilanjana Datta}
\email{n.datta@damtp.cam.ac.uk}
\affiliation{Department of Applied Mathematics and Theoretical Physics, Centre for Mathematical Sciences, University of Cambridge, Cambridge CB3 0WA, United Kingdom}

\begin{abstract}
The precise one-shot characterisation of operational tasks in classical and quantum information theory relies on different forms of smooth entropic quantities. A particularly important connection is between the hypothesis testing relative entropy and the smooth max-relative entropy, which together govern many operational settings.
We first strengthen this connection into a type of equivalence: we show that the hypothesis testing relative entropy is equivalent to a variant of the smooth max-relative entropy based on the information spectrum divergence, which can be alternatively understood as a measured smooth max-relative entropy. 
Furthermore, we improve a fundamental lemma due to Datta and Renner that connects the different variants of the smooth max-relative entropy, introducing a modified proof technique based on matrix geometric means and a tightened gentle measurement lemma.
We use the unveiled connections and tools to strictly improve on previously known one-shot bounds and duality relations between the smooth max-relative entropy and the hypothesis testing relative entropy, establishing provably tight bounds between them. The results then allow us to refine other divergence inequalities, in particular sharpening bounds that connect the max-relative entropy with R\'enyi divergences. 
\end{abstract}

\maketitle

\tableofcontents


\section{Introduction}

Smooth relative entropies~\cite{renner_2005} were defined to address the need for a precise understanding of the performance of various operational protocols in information theory 
in settings beyond the asymptotic i.i.d.~one, where it is assumed that the underlying resources (sources, channels, or entangled states) employed in the protocols are available for asymptotically many uses. 
In this asymptotic setting, unique entropic quantities naturally emerge as the 
optimal rates of information-theoretic tasks. In contrast, when the assumptions of the asymptotic i.i.d.~setting are lifted, i.e.~in the so-called {\em{one-shot information theory}}, several different, inequivalent types of 
smooth entropies are required to fully characterise the 
performance of different tasks. Similarly to how many problems in classical information theory can be broadly divided into two types, packing-type and covering-type problems, the more general setting of quantum information can also be categorised in a similar way. On one side, there are problems that can be connected with hypothesis testing, including quantum channel coding~\cite{wang_2012,buscemi_2010-1,leung_2015,anshu_2019-1,cheng_2023-1}, quantum data compression~\cite{renes_2012,cheng_2023-1}, and quantum resource distillation~\cite{brandao_2010-1,liu_2019,regula_2020}. Due to the underlying hypothesis testing structure, the one-shot performance of these protocols can be related with a quantity known as the \emph{hypothesis testing relative entropy} $D^\ve_H$~\cite{wang_2012,buscemi_2010-1}. 
On the other side, covering-type problems include quantum channel simulation~\cite{berta_2011,fang_2020}, privacy amplification~\cite{renner_2005,tomamichel_2013,shen_2024}, decoupling~\cite{dupuis_2014}, convex split~\cite{anshu_2017}, and quantum resource dilution~\cite{brandao_2010-1,liu_2019}. The one-shot characterisation of such protocols typically relies on a quantity known as the \emph{smooth max-relative entropy} $D^\ve_{\max}$~\cite{renner_2005,datta_2009-2}. Understanding the precise relations between the two smooth relative entropies has thus been an important problem in the one-shot characterisation of quantum information theory~\cite{datta_2013-1,tomamichel_2013,dupuis_2012,anshu_2019,wang_2019}.

Despite their many differences, the hypothesis testing relative entropy and the max-relative entropy are known to be quantitatively related, albeit in a complementary fashion. Informally, the value of $D^\ve_{\max}$ closely matches that of $D^{1-\ve}_H$, up to asymptotically negligible factors~\cite{tomamichel_2013,datta_2013-1}. This means that the two smooth quantities satisfy a  `weak/strong converse duality': understanding the behaviour of $D^\ve_H$ in the small $\ve$ regime (weak converse) closely mirrors the understanding the large $\ve$ (strong converse) behaviour of $D^{\ve}_{\max}$, and vice versa.
This fact was crucial in obtaining results such as asymptotic equipartition theorems for smooth entropic quantities~\cite{tomamichel_2009,brandao_2010,datta_2009,datta_2013-1,tomamichel_2016,lami_2024-2,lami_2024-1} or in characterising the higher-order corrections of optimal rates in quantum information tasks~\cite{tomamichel_2013,datta_2015}.
However, the precise one-shot bounds that connected $D^\ve_{\max}$ and $D^{1-\ve}_H$ in the literature are often not tight, motivating us to look for alternative approaches and stronger links.
\smallskip

\section{Summary of main results}

In this paper, we introduce new techniques for the study of the relations between 
$D^\ve_{\max}$ and $D^{1-\ve}_H$, strengthening both the conceptual and the quantitative connections. {The framework introduced here is strictly stronger than all prior techniques: we strengthen or recover all previously known bounds, and many of our contributions are major improvements over the state of the art, yielding provably tight bounds.}
Our approach employs an 
intermediate  quantity, $\wt{D}^\ve_{\max}$, which is related to the smooth max-relative entropy and was originally introduced in~\cite{datta_2015} as a type of information spectrum divergence~\cite{tomamichel_2013} (see Section~\ref{sec:prelim} for definitions). 
This variant of the max-relative entropy was previously used as a technical tool in many proofs and also found use in characterising quantum differential privacy. Our closer investigation in Section~\ref{sec:equiv} reveals that it is intrinsically connected to both $D^\ve_{\max}$ and $D^{1-\ve}_H$. Specifically, we show that $\wt{D}^\ve_{\max}$ corresponds to a measured variant of the smoothed max-relative entropy (Proposition~\ref{prop:equivalence_measured}), and that it is in fact equivalent to $D^{1-\ve}_H$ in a precise sense: 
in Theorem~\ref{thm:equivalence_DH_Dtilde} 
we prove that, for all pairs of states $\rho$ and $\sigma$ and for all $\e\in (0,1)$, 
\begin{align}
  D^{1-\varepsilon}_H(\rho \| \sigma) = \inf_{\mu\in(0,\ve]}\left[ \widetilde{D}^{\ve-\mu}_{\max} (\rho \| \sigma) + \log\frac1\mu \right],\quad 
  \widetilde{D}^{1-\varepsilon}_{\max}(\rho\|\sigma) = \sup_{\mu \in (0,\varepsilon]} \Big[ D^{\ve-\mu}_H (\rho \| \sigma) - \log \frac1\mu \Big] ,
\end{align}
implying that one can reconstruct the hypothesis testing relative entropy function $D^{1-\ve}_H$ from $\wt{D}^\ve_{\max}$, and vice versa. 

This finding allows us to establish new tight bounds between the hypothesis testing relative entropy and $\wt{D}^\ve_{\max}$ (Lemma~\ref{lem:DH_Dtilde}). To relate the latter with the operationally important quantity $D^\ve_{\max}$, previous works 
relied on an important lemma originally shown by Datta and Renner~\cite{datta_2009-1}, which, however, is not tight for several types of distance measures employed in the smoothing. We introduce an improved proof technique based on the operator geometric mean that leads to a tightening of the previous statement, in particular when the smoothing is over normalised quantum states. Namely, in Theorem~\ref{tightened_DR_lemma} we show that given a state $\rho$ and two positive semi-definite operators $A,Q\geq 0$ with $\Tr Q\leq \ve < 1$, if $\rho$ is approximately dominated by $A$ with `remainder' $Q$, in the sense that the operator inequality $\rho\leq A+Q$ holds, then we can find a smoothing of $\rho$ that is \emph{exactly} dominated by $A$ up to a small rescaling: formally, we can find a normalised state $\rho'$ such that 
\bb
F\big(\rho,\,\rho'\big) \geq 1-\ve\, ,\qquad \frac12 \left\| \rho - \rho' \right\|_1 \leq \sqrt{\ve}\, ,\qquad \rho' \leq \frac{A}{1-\ve}\, .
\ee
We also obtain an improved statement of this result for smoothing over subnormalised quantum states, which involves a tightening of another key tool in quantum information theory, 
the gentle measurement lemma (see Lemma~\ref{tighter_gentle_measurement_lemma}). 
Our proof method is also easy to generalise, which we exemplify with a multi-partite extension of the lemma to a setting of simultaneous state smoothing (Corollary~\ref{cor:joint_smoothing}). 

We apply our results to obtain a set of tight bounds that elucidate the 
duality relations between $D^\ve_{\max}$ and $D^{1-\ve}_H$ in Section~\ref{sec:bounds}. 
{%
We consider both of the commonly used definitions of smoothing in quantum information --- based on  either the trace distance or the purified distance --- and in the process characterise the connections and differences between them. 
Our main result in  Theorem~\ref{cor:wsc} shows that the following relations hold for all states $\rho$ and $\sigma$, all $\ve \in (0,1)$, and all $\mu\in (0,\ve]$: for the trace distance smoothing, we have
\bb\label{eq:intro_dmax_trace}
D_{\max}^{\sqrt{\ve}}(\rho\|\sigma)  + \log \frac1\ve \leq\; & D^{1-\ve}_H(\rho\|\sigma) \leq D_{\max}^{\ve-\mu}(\rho\|\sigma) + \log \frac{1}\mu\,,
\ee
and for the purified smoothing we get
\bb\label{eq:intro_dmax_purified}
D_{\max}^{\sqrt{\ve}}(\rho\|\sigma)  + \log \frac1\ve \leq\; & D^{1-\ve}_H(\rho\|\sigma) \leq D_{\max}^{\sqrt{\ve-\mu}}(\rho\|\sigma) 
+ \log \frac{F_2(1 - \e,\e-\mu)}{\mu^2},
\ee
where 
$F_2(p,q) \coloneqq \left(\sqrt{pq} + \sqrt{(1-p)(1-q)}\right)^2 \leq 1$ denotes the binary fidelity. 
Both of the lower bounds on $D^{1-\ve}_H$ are a significant improvement over previously known one-shot relations, and our purified distance upper bound in~\eqref{eq:intro_dmax_purified} also refines all known results of this type.
Crucially, all of the inequalities are tight: none of the $\ve$ dependencies can be improved in general, with the lower bounds tight in the case $\sigma=\rho$, the upper bound of~\eqref{eq:intro_dmax_trace} tight for classical (commuting) states, and the upper bound of~\eqref{eq:intro_dmax_purified} tight for pure states.
}%

Going further, in Section~\ref{sec:bounds_renyi} we employ our techniques to strengthen also the bounds connecting the smooth max-relative entropy with the R\'enyi divergences, leading also to a bound for the hypothesis testing relative entropy. To wit, in Corollaries~\ref{cor:var-mrel} and~\ref{cor:renyi_lower} we show that, for all $\ve\in (0,1)$, all $\alpha>1$, and all $\beta \in (0,1)$,
\bb
D_{\max}^\ve(\rho \| \sigma) &\leq \wt{D}_{\alpha} (\rho \| \sigma) + \frac{1}{\alpha-1} \log \frac{1}{\ve^2}\, ,\\
D^\ve_H(\rho\|\sigma) &\geq D_\beta(\rho \| \sigma) - \frac{\beta}{1-\beta} \log \frac{1}{\ve} + \log \frac{1}{1-\ve}\,,
\ee
where $\wt{D}_\alpha$ and $D_\beta$ denote, respectively, the sandwiched R\'enyi relative entropy (given by~\eqref{sandwiched_Renyi}) and the Petz--R\'enyi relative entropy~\eqref{Petz_Renyi}.
The bounds obtained in this way give tight constraints on the asymptotic expansion of several of the smooth divergence variants discussed here.

We also obtain other auxiliary results which may be of independent interest: {(i)~improved  connections between the hypothesis testing relative entropy $D^\ve_H$ and the information spectrum divergence $D^\ve_s$ (Proposition~\ref{lem:info_spec});} (ii)~a tightening of the so-called quantum substate theorem (Corollary~\ref{cor:qsstate}); (iii)~a variant (Theorem~\ref{thm:Frenkel-extn}) of  a recent important result, namely, Frenkel's integral representation for the Umegaki relative entropy, but with the integrand being a function of $\wt{D}^\ve_{\max}$, and (iv)~a weak/strong converse duality relation between 
$D_H^{1-\ve}$ and the Hilbert projective metric (Corollary~\ref{cor:hilbert_wsc}).


\section{Smooth divergences}\label{sec:prelim}


\subsection{Hypothesis testing relative entropy and max-relative entropy}

We use Greek letters ($\rho, \sigma$) to denote quantum states, which we understand to be positive semidefinite operators of trace one acting on a finite-dimensional Hilbert space. We will sometimes specialise our results to classical probability distributions, which are denoted by Latin letters ($p,q$); we can equivalently interpret them as commuting quantum states --- that is, states that are diagonal in a common basis. 
All logarithms are to base 2 unless otherwise stated.

Smooth divergences are broadly understood as divergences (relative entropies) between quantum states that incorporate the allowance for some form of error or uncertainty about the states under consideration, quantified by a smoothing parameter $\ve$.

Perhaps the most fundamental divergence 
in one-shot quantum information theory is the \deff{hypothesis testing relative entropy}, defined for two quantum states $\rho$ and $\sigma$ as~\cite{wang_2012,buscemi_2010-1}
\begin{equation}\begin{aligned}\label{eq:dh_def}
    D^\ve_H (\rho \| \sigma) \coloneqq - \log \inf \lset \Tr M \sigma \bar 0 \leq M \leq \id,\; \Tr (\id - M) \rho \leq \ve \rset,
\end{aligned}\end{equation}
where $\ve \in [0,1]$.
Equivalently, $D^\ve_H (\rho \| \sigma) = - \log \beta_\ve(\rho, \sigma)$, where $\beta_\ve(\rho, \sigma)$ denotes the optimal probability of the type II error of hypothesis testing between $\rho$ and $\sigma$ when the type I error probability is constrained to be at most $\ve$.
The divergence is conceptually very closely related to smooth relative entropies, although it is defined slightly differently than conventional smooth divergences~\cite{renner_2005} --- rather than employing an optimisation over a neighbourhood of states, it can be regarded a smoothed variant of the Petz--R\'enyi divergence of order 0 (min-relative entropy~\cite{datta_2009-2}) under a notion of `operator smoothing'~\cite{buscemi_2010-1}.

Another fundamental quantity that we focus on here is the \deff{max-relative entropy}~\cite{datta_2009-2}
\begin{equation}\begin{aligned}\label{eq:dmax_def}
    D_{\max} (\rho \| \sigma) \coloneqq 
    \log \inf \lset \lambda \geq 0 \bar \rho \leq \lambda \sigma \rset,
\end{aligned}\end{equation}
{where $\inf \emptyset = \infty$ and we take $\log(\infty) \coloneqq \infty$ for consistency.} 
This can be understood as the sandwiched R\'enyi divergence $\wt{D}_\alpha$ of order $\infty$~\cite{muller-lennert_2013}. 
The smooth variant of this quantity is defined not by evaluating $D_{\max}(\rho \| \sigma)$ exactly, but instead by 
minimising
over all states $\rho' \approx_{\ve} \rho$ in an $\ve$-ball around $\rho$~\cite{renner_2005,datta_2009-2}. This definition crucially depends on how exactly we define the smoothing neighbourhood, and in particular on how we quantify the distance between quantum states.

Two of the most fundamental measures of distance are the \deff{trace distance} (total variation distance) $\frac12 \norm{\rho - \rho'}{1}$, where $\norm{X}{1} = \Tr \sqrt{X^\dagger X}$ is the Schatten 1-norm, as well as the \deff{purified distance} $P(\rho, \rho') \coloneqq \sqrt{1 - F(\rho,\rho')}$, where $F(\rho,\rho')$ denotes the Uhlmann fidelity
\begin{equation}\begin{aligned}\label{eq:fidelity_def}
    F(\rho,\rho') &\coloneqq \norm{\sqrt{\rho\vphantom{T}}\sqrt{\rho'}}{1}^{\,2} = \left( \Tr\sqrt{\sqrt{\rho}\rho'\sqrt{\rho}} \right)^2.
\end{aligned}\end{equation}
The definition of the distance measures is often extended to cover the case where $\rho'$ is not necessarily a normalised quantum state, but rather a \emph{subnormalised} positive operator with $\Tr \rho' \leq 1$. As long as $\rho$ is normalised, the definition of fidelity in Eq.~\eqref{eq:fidelity_def} does not change~\cite{tomamichel_2016}; however, the form of the trace distance needs to be adjusted to account for subnormalised states.
For $\rho$ and $\rho'$ with $\Tr \rho = 1$ and $\Tr \rho' \leq 1$, we define the \deff{generalised trace distance} as~\cite{tomamichel_2016}
\begin{equation}\begin{aligned}
    \norm{\rho - \rho'}{+} &\coloneqq \frac12\norm{\rho - \rho'}{1} + \frac12 \Tr(\rho-\rho')
    = \Tr (\rho - \rho')_+\, ,
\end{aligned}\end{equation}
where $\Tr(\cdot)_+$ denotes the trace of the positive part of a Hermitian operator. 
Using a variational characterisation of $\Tr(\cdot)_+$, we can rewrite this as
\begin{equation}\begin{aligned}\label{eq:trace_variational}
    \norm{\rho - \rho'}{+} = \max_{0 \leq M \leq \id} \Tr \left[ M (\rho - \rho') \right] =  \min \lset \Tr X \bar \rho - \rho' \leq X,\; X \geq 0 \rset,
\end{aligned}\end{equation}
where the formulation as a maximum neatly expresses the quantum Neyman--Pearson lemma~\cite{HELSTROM}, while the dual formulation as a minimum will often find use in this paper.

We then define the \deff{smooth max-relative entropy} as
\begin{equation}\begin{aligned}
    D^{\smoothing{\Delta}{\ve}}_{\max} (\rho \| \sigma) \coloneqq \inf_{\rho' \in \B^\ve_\Delta(\rho)} D_{\max}(\rho' \| \sigma),
\end{aligned}\end{equation}
where $\Delta$ denotes one of the following choices of smoothing:


\hspace*{-10pt}
\begin{tblr}{
  colspec = {m{1.8cm} Q[l,wd=3.8cm] l},
  stretch = 1.15
}
\toprule
Notation & Smoothing metric & Definition of smoothing ball\\
\midrule
$\sTnorm$ & Trace distance, normalised &
$\B^\ve_{\sTnorm}(\rho) \coloneqq \lset \rho' \bar \frac12\norm{\rho - \rho'}{1} \leq \ve,\; \rho' \geq 0,\; \Tr \rho' = 1 \rset$\\

$\sTsub$ & Trace distance, subnormalised &
$\B^\ve_{\sTsub}(\rho) \coloneqq \lset \rho' \bar \norm{\rho - \rho'}{+} \leq \ve,\; \rho' \geq 0,\; \Tr \rho' \leq 1 \rset$\\

$\sPnorm$ & Purified distance, normalised &
$\B^\ve_{\sPnorm}(\rho) \coloneqq \lset \rho' \bar \sqrt{1-F(\rho,\rho')} \leq \ve,\; \rho' \geq 0,\; \Tr \rho' = 1 \rset$\\

$\sPsub$ & Purified distance, subnormalised &
$\B^\ve_{\sPsub}(\rho) \coloneqq \lset \rho' \bar \sqrt{1-F(\rho,\rho')} \leq \ve,\; \rho' \geq 0,\; \Tr \rho' \leq 1 \rset$\\
\bottomrule
\end{tblr}\\

To relate the two types of smoothing based on trace and purified distance, from the Fuchs--van de Graaf inequalities~\cite{fuchs_1999,tomamichel_2016}
\begin{equation}\begin{aligned}\label{eq:fvdg}
    1 - \sqrt{F(\rho, \rhos)} \leq \|\rho - \rhos\|_+ \leq \sqrt{1- F(\rho,\rhos)},
\end{aligned}\end{equation}
one derives immediately that
\bb
\Dmax{P}{n}[\sqrt{\ve(2-\ve)}](\rho\|\sigma) \leq \Dmax{T}{n}(\rho\|\sigma) \leq \Dmax{P}{n}(\rho\|\sigma)
\label{standard_vs_infidelity}
\ee
and likewise for the subnormalised variants. 

The smoothed quantities $D^{\ve}_H$ and $D^\ve_{\max}$ are not a priori related. A fundamental connection between them arises from the fact that their difference becomes neglibile when one considers many i.i.d.\ copies of quantum states, and it is known that~\cite{hiai_1991,ogawa_2000,tomamichel_2009}
\begin{equation}\begin{aligned}
    \lim_{n\to\infty} \frac1n D^{\ve}_H(\rho^{\otimes n} \| \sigma^{\otimes n}) = \lim_{n\to\infty} \frac1n \Dmax{P}{n}(\rho^{\otimes n} \| \sigma^{\otimes n}) = \lim_{n\to\infty} \frac1n \Dmax{T}{n}(\rho^{\otimes n} \| \sigma^{\otimes n}) = D(\rho \| \sigma)
\end{aligned}\end{equation}
for all $\ve \in (0,1)$, where 
\bb
D(\rho\|\sigma) \coloneqq \Tr \rho (\log \rho - \log \sigma)
\label{Umegaki}
\ee
denotes the (Umegaki) quantum relative entropy. In~\eqref{Umegaki}, one sets $D(\rho\|\sigma)=\infty$ if $\supp \rho \not\subseteq \supp\sigma$. 
This was later tightened to a one-shot relation in~\cite{tomamichel_2013}, namely
\begin{equation}\begin{aligned}\label{eq:first_weakstrong}
    \Dmax{P}{sub}[\sqrt{\ve}](\rho\|\sigma) - \log \frac{\left|\operatorname{spec}(\sigma)\right|}{\ve}  &\leq D^{1-\ve}_H(\rho\|\sigma) \leq \Dmax{P}{sub}[\sqrt{\ve-\mu}](\rho\|\sigma) + \log \frac{27 (1-\ve+\mu)}{\mu^3}
\end{aligned}\end{equation}
{%
for all $\ve \in (0,1)$ and all $\mu\in (0,\ve]$. 
This showed that the values of the two divergences can be quantitatively related in complementary error regimes --- a property that we refer to as the `weak/strong converse duality' between $D^\ve_H$ and $D^\ve_{\max}$. We note that this terminology is not necessarily meant to imply a rigorous connection with an operational task, but rather to merely evoke the intuitive idea that the weak converse regime is concerned with small $\ve$ while the strong converse regime with large $\ve$. 
}%
These bounds were subsequently improved in many works~\cite{datta_2013-1,dupuis_2012,anshu_2019}, but the problem of finding tight one-shot relations and bounds remained open.


\subsection[\texorpdfstring%
{Modified max-relative entropy ${\wt{D}^\ve_{\max}}$}%
{Modified max-relative entropy}]{Modified max-relative entropy $\boldsymbol{\wt{D}^\ve_{\max}}$}
\label{Sec:newdmax}

An important quantity in our approach will be a variant of the smooth max-relative entropy, formalised by Datta and Leditzky~\cite{datta_2015} as
\begin{equation}\begin{aligned}
    \wt{D}^\ve_{\max} (\rho \| \sigma) \coloneqq \log \inf \lset \lambda \geq 0 \bar \Tr (\rho - \lambda \sigma)_+ \leq \ve \rset.
    \label{Datta_Leditzky}
\end{aligned}\end{equation}

In fact, it was first introduced as an alternative to the information spectrum divergence $D_s^\ve$~\cite{tomamichel_2013} and denoted $\overline{D}_s^{\,\ve}$ in~\cite{datta_2015}. Here we prefer to relate it to the max-relative entropy, the connection with which becomes clearer once we use the variational form of $\Tr(\cdot)_+$ to write~\cite{nuradha_2024}
\begin{equation}\begin{aligned}\label{eq:Dtilde_def}
    \wt{D}^\ve_{\max} (\rho \| \sigma) = \log \inf \lset \lambda \bar \rho \leq \lambda \sigma + Q,\; Q \geq 0,\; \Tr Q \leq \ve \rset,
\end{aligned}\end{equation}
which more closely resembles the optimisation problem that defines $D_{\max}$. Indeed, we notice that $\wt{D}^0_{\max} (\rho\|\sigma) = D_{\max}(\rho \| \sigma)$, which justifies regarding $\wt{D}^\ve_{\max}$ as a smoothed variant of $D_{\max}$.

{As we show in detail in Appendix~\ref{app:continuity}, for any two fixed states $\rho,\sigma$ the function $\e\mapsto \widetilde{D}_{\max}^\e(\rho\|\sigma)$ is continuous and strictly decreasing on the interval $\e\in \big(1-\Tr \rho \Pi_\sigma, 1\big]$, where $\Pi_\sigma$ denotes the projector onto the support of $\sigma$.}

The modified max-relative entropy $\wt{D}_{\max}^\ve$ made implicit appearances in many early works on quantum hypothesis testing already~\cite{ogawa_2000,nagaoka_2007,brandao_2010,mosonyi_2015} through its connection with the information spectrum method~\cite{nagaoka_2007,datta_2009-1}. 
{%
By Eq.~\eqref{Datta_Leditzky}, it can be understood as {a generalised} inverse function of the hockey-stick divergence $E_\lambda (\rho \| \sigma) \coloneqq \Tr (\rho - \lambda \sigma)_+$, which has also found applications in various aspects of one-shot information theory~\cite{sharma_2013,liu_2017-1,yang_2019,hirche_2023,Hirche2023}.  
 $\wt{D}_{\max}^\ve$} was first formalised as a divergence in~\cite{datta_2015}, where it was used to study second-order expansions for optimal rates of certain quantum information-theoretic tasks in the i.i.d.~setting.
 Some of its basic properties, including the data processing inequality, were established there. 
The classical analogue of this quantity independently found use in the study of differential privacy~\cite{dwork_2010,dwork_2014}, and a generalisation of these results later also led to $\wt{D}^\ve_{\max}$ playing a role in characterising quantum differential privacy~\cite{nuradha_2024}.\footnote{The use of $\wt{D}^\ve_{\max}$ in the context of quantum differential privacy can in fact be traced back to the earlier works~\cite{zhou_2017,hirche_2023}. In particular, some results of \cite{hirche_2023} are stated in terms of an inverse function of $E_\lambda$. 
The same quantity implicitly appeared in~\cite{zhou_2017}, although it is not immediately clear that the definition there corresponds to $\wt{D}^\ve_{\max}$. We clarify the equivalence of the different definitions in Appendix~\ref{app:definitions}.} A conditional entropy variant of the divergence also made appearances in the study of classical privacy amplification~\cite{renes_2018,abdelhadi_2020}, where it was connected with the smooth min-entropy~\cite{renner_2004}.

Although $\wt{D}^\ve_{\max}$ was employed mostly as a technical tool in one-shot quantum information theory, we will show that the quantity has many precise connections with both the hypothesis testing relative entropy $D^\ve_H$ and with the standard smoothed max-relative entropy $\Dmax{T}{n}$ that were not previously known, ultimately leading to tighter bounds between all of these quantities.


{%
\section{On normalised and subnormalised smoothing}\label{sec:subnormalised}
}


{%
As the first step of our investigation, we will 
clarify a basic property of the smooth divergences that will help consolidate the many possible choices of smoothing encountered in the literature.

A natural question in the definition of the smooth max-relative entropy is why the ostensibly less physical subnormalised states are considered among the smoothing variants. 
In operational settings, it may indeed seem desirable to optimise only over normalised states, as these are the only states that can result from the application of a deterministic state transformation (quantum channel).
However, some technical benefits can be gained by relaxing the smoothing to subnormalised states --- notably, the invariance of the resulting conditional entropic quantities under isometries, which leads to useful duality relations~\cite{tomamichel_2010,tomamichel_2016} --- making this a frequently made choice in many works on one-shot quantum information theory.

We can show, however, that the two smoothing notions are equivalent
whenever the normalised smooth max-relative entropy is non-zero, or, equivalently, when the subnormalised quantity takes positive values. Naturally, most of the questions about the behaviour of these quantities is concerned with regimes in which they are larger than zero, and hence it is useful to realise that the dichotomy between the normalised and subnormalised definitions is effectively irrelevant in such cases.

\begin{boxed}
\begin{lemma}\label{lem:subnormalised_normalised_main}
For any two quantum states $\rho$ and $\sigma$ and any $\ve \in [0,1]$, it holds that
\begin{equation}\begin{aligned}
    \Dmax{T}{n} (\rho \| \sigma) &= \max \Big\{ \Dmax{T}{sub} (\rho\|\sigma),\, 0 \Big\},\\
    \Dmax{P}{n} (\rho \| \sigma) &= \max \left\{ \Dmax{P}{sub} (\rho\|\sigma),\, 0 \right\}.
\end{aligned}\end{equation}
As a consequence, $\Dmax{T}{n} (\rho \| \sigma) = \Dmax{T}{sub} (\rho \| \sigma)$ if and only if $\ve \leq \frac12\norm{\rho-\sigma}{1}$, and $\Dmax{P}{n} (\rho \| \sigma) = \Dmax{P}{sub} (\rho \| \sigma)$ if and only if $\ve \leq P(\rho,\sigma)$.
\end{lemma}
\end{boxed}

\begin{proof}
Note first that, since $D_{\max}(\rho \| \sigma) \geq 0$ for any normalised operators $\rho$ and $\sigma$, it must be the case that $\Dmax{T}{n}(\rho\|\sigma)$ and $\Dmax{P}{n}(\rho\|\sigma)$ are non-negative. 
Now, if $\Dmax{T}{sub}(\rho\|\sigma) < 0$, then there exists a subnormalised state $\rho'$ such that $\rho' < \sigma$, where we assumed that $\sigma > 0$ since we can always restrict the space down to the support of $\sigma$ without loss of generality. 
This implies that $\rho - \sigma < \rho - \rho' \leq (\rho - \rho')_+$. But then the assumption that $\norm{\rho-\rho'}{+}\leq\ve$ means that $\frac12\norm{\rho-\sigma}{1} = \Tr(\rho-\sigma)_+ < \Tr(\rho - \rho')_+ \leq \ve$, 
and so $\Dmax{T}{n}(\rho\|\sigma) \leq D_{\max}(\sigma\|\sigma) = 0$. In the case of $\Dmax{P}{sub}(\rho \| \sigma)$ we similarly have that $\rho' < \sigma$ for some subnormalised state such that $F(\rho, \rho') \geq 1-\ve^2$. Using the strict operator monotonicity of the square root (see e.g.~\cite[Proposition~1.2.9]{bhatia_2007}), we get that
\begin{equation}\begin{aligned}
    F(\rho, \sigma) = \left( \Tr \sqrt{\sqrt{\rho}\sigma\sqrt{\rho}} \right)^2 > \left( \Tr \sqrt{\sqrt{\rho}\rho'\sqrt{\rho}} \right)^2 = F(\rho, \rho'),
\end{aligned}\end{equation}
so that again $\Dmax{P}{n}(\rho\|\sigma) \leq D_{\max}(\sigma\|\sigma) = 0$.

We can in fact make a stronger statement here, namely, $\Dmax{T}{sub}(\rho\|\sigma) < 0 \iff \ve > \frac12\norm{\rho-\sigma}{1}$, and analogously for the purified distance. 
We have already shown the $\Rightarrow$ implication. To see the other direction, assume that $\ve > \frac12\norm{\rho-\sigma}{1}$. Clearly, $\Dmax{T}{n}(\rho\|\sigma) = 0$ because $\sigma \in \B^\ve_{\sTnorm}(\rho)$, so $\Dmax{T}{sub}(\rho\|\sigma)$ can be at most zero. But if it were exactly zero, then we would have that $\rho' \leq \sigma$ for some subnormalised state $\rho'$ with $\norm{\rho-\rho'}{+}\leq \ve$, and hence $\frac12\norm{\rho-\sigma}{1} \leq \Tr(\rho - \rho')_+ \leq \ve$, contradicting our assumption. Hence $\Dmax{T}{sub}(\rho\|\sigma) < 0$.

We will now show that if $\Dmax{T}{sub}(\rho\|\sigma) \geq 0$, then $\Dmax{T}{sub}(\rho\|\sigma) = \Dmax{T}{n}(\rho\|\sigma)$. Let $\rho'$ be a feasible subnormalised state such that $\rho' \leq \lambda \sigma$ and $\norm{\rho-\rho'}{+}\leq\ve$ for some $\lambda \geq 1$. To avoid trivial cases we will assume that $\Tr \rho' < 1$. Denoting by $\{\ket{i}\}_i$ an orthonormal basis that diagonalises the operator $\lambda \sigma - \rho'$, let $\rho'_{ij} \coloneqq \braket{i|\rho'|j}$ and $\sigma_{ij} \coloneqq \braket{i|\sigma|j}$. Define the sets of indices
\begin{equation}\begin{aligned}
    S \coloneqq \lset i \bar \sigma_{ii} < \rho'_{ii} \leq \lambda \sigma_{ii} \rset, \qquad S^\perp \coloneqq \lset i \bar \sigma_{ii} \geq \rho'_{ii} \rset.
\end{aligned}\end{equation}
For some value of $\mu \in [0,1]$ to be fixed later, we now define
\begin{align}
    \rho'' &\coloneqq \rho' + \mu \sum_{i \in S^\perp} (\sigma_{ii} - \rho'_{ii}) \proj{i} \nonumber \\
    &\hphantom{:}\textleq{(i)} \rho' + \sum_{i \in S^\perp} (\sigma_{ii} - \rho'_{ii}) \proj{i} \nonumber \\
     &\hphantom{:}= \sum_{i \in S } \rho'_{ii} \proj{i} + \sum_{i\neq j} \rho'_{ij} \ketbraa{i}{j} + \sum_{i \in S^\perp} \sigma_{ii}\proj{i} \label{eq:rho_double_prime} \\
     &\hphantom{:}\textleq{(ii)} \sum_{i\in S} \lambda \sigma_{ii} \proj{i} + \sum_{i\neq j} \lambda \sigma_{ij} \ketbraa{i}{j}+ \sum_{i \in S^\perp} \sigma_{ii}\proj{i} \nonumber \\
     &\hphantom{:}\textleq{(iii)} \lambda \sigma, \nonumber
\end{align}
where: (i) follows since $\mu \leq 1$ and $\sigma_{ii} \geq \rho'_{ii}$ for all $i \in S^\perp$ by definition; (ii) follows as $\rho'_{ii} \leq \lambda \sigma_{ii}$ for all $i \in S$, and $\lambda \sigma_{ij} = \rho'_{ij}$ because $\{\ket{i}\}_i$ diagonalises $\lambda \sigma - \rho'$; finally, (iii) is simply since $\lambda \geq 1$.

Observe that there must exist at least one element in $S^\perp$ such that $\rho'_{ii} < \sigma_{ii}$, as otherwise our assumption that $\Tr \rho' < 1$ would be contradicted. We can then always choose a value of $\mu$ so that $\Tr \rho'' = 1$, namely
\begin{equation}\begin{aligned}
    \mu \coloneqq \frac{ 1 - \Tr \rho' }{ \sum_{i\in S^\perp} (\sigma_{ii}-\rho'_{ii})}.
\end{aligned}\end{equation}
The fact that $\mu \leq 1$ can be seen by noticing that
\begin{equation}\begin{aligned}
   1 - \Tr \rho' &= \!\!\sum_{i \in S \cup S^\perp} (\sigma_{ii} - \rho'_{ii}) \leq \sum_{i \in S^\perp} (\sigma_{ii} - \rho'_{ii})
\end{aligned}\end{equation}
where we used that $\sigma_{ii} < \rho'_{ii}$ for all $i \in S$.

Now, clearly, $\rho'' \geq \rho'$. This implies that $\rho - \rho'' \leq \rho - \rho'$ and hence that
\begin{equation}\begin{aligned}
    \frac12 \norm{\rho - \rho''}{1} = \Tr ( \rho - \rho'')_+ \leq \Tr (\rho - \rho')_+ \leq \ve.
\end{aligned}\end{equation}
Thus $\Dmax{T}{n} (\rho \| \sigma) \leq D_{\max}(\rho'' \| \sigma) \leq \log \lambda$ by~\eqref{eq:rho_double_prime}. Since this holds for all feasible $\lambda$, we have that $\Dmax{T}{n} (\rho \| \sigma) \leq \Dmax{T}{sub}(\rho\|\sigma)$ and thus the two quantities must be equal.

The proof for the purified distance is completely analogous. We now start with a subnormalised state $\rho'$ such that $\sqrt{1-F(\rho,\rho')}\leq \ve$
and using again the operator monotonicity of the square root we obtain
\begin{equation}\begin{aligned}
     F(\rho, \rho'') = \left( \Tr \sqrt{\sqrt{\rho}\rho''\sqrt{\rho}} \right)^2 \geq \left( \Tr \sqrt{\sqrt{\rho}\rho'\sqrt{\rho}} \right)^2 = F(\rho, \rho'),
\end{aligned}\end{equation}
from which we get that $\Dmax{P}{n} (\rho \| \sigma) \leq D_{\max}(\rho'' \| \sigma) \leq \log \lambda$.
\end{proof}

Another way to relate the normalised and subnormalised quantities without having to distinguish the cases when the subnormalised smoothed divergences take negative values is to take any subnormalised state $\rho'$ and normalise it as $\rho'' \coloneqq \frac{\rho'}{\Tr \rho'}$; it is not difficult to show that $\frac12\norm{\rho-\rho''}{1} \leq \norm{\rho-\rho'}{+}$ and $F(\rho,\rho'') \geq F(\rho,\rho')$, meaning that $\rho''$ is a feasible state for the normalised smooth $D_{\max}$. Accounting for the renormalisation factor, $\Tr \rho'$, gives
\begin{equation}\begin{aligned}
    \Dmax{T}{sub} (\rho \| \sigma) \leq \Dmax{T}{n} (\rho \| \sigma) \leq \Dmax{T}{sub}(\rho \| \sigma) + \log \frac{1}{1-\ve},
\end{aligned}\end{equation}
where we used that $\Tr \rho' \geq 1-\ve$ for any $\rho'$ with $\norm{\rho-\rho'}{+} \leq \ve$. The purified distance relation is analogous.

}%


\section{Equivalences}\label{sec:equiv}

{%
\subsection[Classical \texorpdfstring{$\wt{D}^\ve_{\max}$ and $D^\ve_{\max}$}{modified max-relative entropy}]{Classical $\boldsymbol{\wt{D}^\ve_{\max}}$ and $\boldsymbol{{D}^\ve_{\max}}$}

For classical systems, there is in fact a complete equivalence between the modified max-relative entropy $\wt{D}^\ve_{\max}$ and the standard smooth max-relative entropy $D^\ve_{\max}$. 
This was first noticed in~\cite{dwork_2010,dwork_2014}, where a classical definition of $\wt{D}^\ve_{\max}$ was introduced in the context of characterising differential privacy.
\begin{boxed}
\begin{lemma}[{\cite{dwork_2014,abdelhadi_2020}}]\label{lem:classical_dmax_equivalence}
For all classical probability distributions $p,q$ (or commuting quantum states), we have
\begin{equation}\begin{aligned}
    \wt{D}^\ve_{\max}(p \| q) = \Dmax{T}{sub} (p\|q).
\end{aligned}\end{equation}
As a consequence,
\begin{equation}\begin{aligned}\label{eq:subnorm_classical}
    \Dmax{T}{n} (p \| q) = \max \left\{ \wt{D}^\ve_{\max}( p \| q),\, 0 \right\}.
\end{aligned}\end{equation}
\end{lemma}
\end{boxed}
The classical result in~\eqref{eq:subnorm_classical} was first shown in~\cite[Lemma~3.17]{dwork_2014} (see also~\cite{dwork_2010}), without studying the subnormalised smoothing variant. It was previously claimed that~\eqref{eq:subnorm_classical} holds also in the quantum case~\cite{zhou_2017}, but the proof contains a gap~\cite{hirche_2023} (and indeed this claim can be easily verified to be false numerically; see also Appendix~\ref{app:definitions}). The classical equality $\Dmax{T}{sub}  = \wt{D}^\ve_{\max}$ was also shown in~\cite[Proposition 2]{abdelhadi_2020} (see also~\cite{hayashi_2016-2,renes_2018}), without connecting it to the normalised $\Dmax{T}{n}$. For completeness, we give a proof here, essentially following~\cite{abdelhadi_2020}. 

\begin{proof}[Proof of Lemma~\ref{lem:classical_dmax_equivalence}]
We will understand $p$ and $q$ as diagonal density operators $p = \sum_i p_i \proj{i}$, $q = \sum_i q_i \proj{i}$ in some orthonormal basis. 
Consider then any feasible solution for $\wt{D}^\ve_{\max}(p\|q)$, that is, $p \leq \lambda q + Q$ for some positive semidefinite operator $Q$ with $\Tr Q \leq \ve$. 
Without loss of generality, $Q$ can be taken to be diagonal in the basis $\{\ket{i}\}$, since we can simply dephase any feasible $Q$ in this basis while maintaining its feasibility. 
Furthermore, we can in fact assume that $p \geq Q$: should it be the case that $p_i < Q_i$ for some $i$, we can define $Q' \leq p$ through $Q'_i \coloneqq \min\{ Q_i, p_i\}$, which is feasible with the same objective value $\lambda$ as
\begin{equation}\begin{aligned}
    \lambda q_i \geq \max \{ p_i - Q_i, 0 \} = p_i - Q'_i \quad \forall {i}.
\end{aligned}\end{equation}
But then $p' \coloneqq p - Q \leq \lambda q$ is a subnormalised state such that $\norm{p-p'}{+} = \norm{Q}{+} = \Tr Q \leq \ve$,
implying that 
\begin{equation}\begin{aligned}
    \Dmax{T}{sub} (p\|q) \leq D_{\max}(p' \| q) \leq \log \lambda.
\end{aligned}\end{equation}
Since this holds for any feasible solution $\lambda$, it follows that $\Dmax{T}{sub}(p\|q) \leq \wt{D}_{\max}^\ve(p\|q)$. 
The opposite inequality between the quantities follows easily by noting that any feasible $p'$ such that $p' \leq \lambda q$ and $\Tr (p -p')_+ \leq \ve$ gives a feasible solution for $\wt{D}_{\max}^\ve(p\|q)$ through $p \leq p' + \Tr(p-p')_+ \leq \lambda q + \Tr(p-p')_+$.  
The quantities thus must be equal, and Eq.~\eqref{eq:subnorm_classical} then follows by Lemma~\ref{lem:subnormalised_normalised_main}.
\end{proof}
}%


\subsection[\texorpdfstring{${\wt{D}^\ve_{\max}}$ as measured ${D^\ve_{\max}}$}{Modified max-relative entropy as measured smooth max-relative entropy}]{$\boldsymbol{\wt{D}^\ve_{\max}}$ as measured $\boldsymbol{D^\ve_{\max}}$}

A natural way to define quantum divergences from classical ones is to first perform a measurement and then compute the classical divergence of the resulting probability distributions; optimising over all measurements then gives a well-defined quantum relative entropy~\cite{donald_1986,hiai_1991,berta_2017-1}. Following this idea, we define the \deff{measured smooth max-relative entropy} as
\begin{equation}\begin{aligned}
    D^{\smoothing{\Delta}{\ve}}_{\max,\, \MM} (\rho \| \sigma) \coloneqq \sup \lset D^{\smoothing{\Delta}{\ve}}_{\max} \left( p_{\rho,M} \| p_{\sigma,M} \right) \bar M = (M_i)_{i=1}^n \in \MM,\; n \in \NN_+ \rset,
    \label{eq-17}
\end{aligned}\end{equation}
where the choices of distance $\Delta$ are as in Section~\ref{sec:prelim}, $p_{\rho,M}$ denotes the probability distribution of the measurement outcomes, $p_{\rho,M}(i) = \Tr M_i \rho$, and $\MM$ is the set of all finite-outcome quantum measurements.

We now show that $\wt{D}^\ve_{\max}$ can be understood precisely as a measured variant of the smooth max-relative entropy $\Dmax{T}{sub}$, giving the former quantity a new interpretation and more closely connecting the two divergences.

\begin{boxed}
\begin{proposition}\label{prop:equivalence_measured}
For all quantum states $\rho$ and $\sigma$ and for all $\ve \in [0,1]$ it holds that
\begin{align}\label{eq:equivalence_measured}
\wt{D}^{\ve}_{\max} (\rho \| \sigma) = D^{\s{T}{sub}}_{\max,\,\MM} (\rho \| \sigma) = \wt{D}^\ve_{\max,\,\MM} (\rho \| \sigma).
\end{align}
It also follows as a consequence that $D^{\s{T}{n}}_{\max,\,\MM} (\rho \| \sigma) = \max \Big\{ \wt{D}^{\ve}_{\max} (\rho \| \sigma) ,\, 0 \Big\}$.
\end{proposition}
\end{boxed}

\begin{proof}
By the data processing inequality for $\wt{D}^\ve_{\max}$, we have that $\wt{D}^\ve_{\max}(\rho \| \sigma) \geq \wt{D}^{\ve}_{\max} (p_{\rho, M} \| p_{\sigma, M})$ for any measurement $M$, and hence
\begin{equation}\begin{aligned}
    \wt{D}^\ve_{\max} (\rho \| \sigma) &\geq \sup_{M \in \MM}  \wt{D}^{\ve}_{\max} (p_{\rho, M} \| p_{\sigma, M}) \\
    &= D^{\s{T}{sub}}_{\max,\,\MM} (\rho \| \sigma)
\end{aligned}\end{equation}
where the last equality is by Lemma~\ref{lem:classical_dmax_equivalence}, using the fact that $p_{\rho, M}$ and $p_{\sigma, M}$ are classical probability distributions and so $\wt{D}^\ve_{\max} (p_{\rho, M} \| p_{\sigma, M}) = D^{\s{T}{sub}}_{\max,\,\MM} (p_{\rho, M} \| p_{\sigma, M})$.

To show the opposite inequality, we begin by using convex duality to write (see also~\cite{nuradha_2024})
\begin{equation} \begin{aligned} \label{eq:Dtilde_dual}
    \wt{D}^\ve_{\max} (\rho \| \sigma) &= \log \inf \lset \lambda \bar \rho \leq \lambda \sigma + Q,\; Q \geq 0,\; \Tr Q \leq \ve \rset 
    \\
    &= \log \sup \lset \Tr W \rho - y \ve \bar 0 \leq W \leq y \id,\; \Tr W \sigma = 1,\; y \geq 0 \rset\\
     &= \log \sup \lset \frac{\Tr W' \rho - \ve}{\Tr W' \sigma} \bar 0 \leq W' \leq \id \rset,
\end{aligned} \end{equation}
where the second equality follows by strong Lagrange duality, noting that $W = \id$, $y = 1 + \delta$ is a strictly feasible solution pair for any $\delta > 0$,
and in last line we made the change of variables $\frac{1}{y} W \mapsto W'$. 
Let us then take any $W'$ feasible for the optimisation on the rightmost side of~\eqref{eq:Dtilde_dual}. Choosing $M' = (W', \id - W')$ we get, using the definition of the measured smooth max-relative entropy in~\eqref{eq-17} and the classical equivalence between $\Dmax{T}{s}$ and $\wt{D}_{\max}^\ve$ shown in Lemma~\ref{lem:classical_dmax_equivalence},
\begin{equation}\begin{aligned}
    D^{\s{T}{sub}}_{\max,\,\MM} (\rho \| \sigma) &\geq 
     D^{\s{T}{sub}}_{\max} (p_{\rho, M'} \| p_{\sigma, M'})\\
    &=
    \wt{D}^{\ve}_{\max} (p_{\rho, M'} \| p_{\sigma, M'})\\
    &= \log \sup \lset \frac{\Tr V p_{\rho, M'} - \ve}{\Tr V p_{\sigma, M'}} \bar 0 \leq V \leq \id \rset\\
    &\geq \log \sup \frac{\Tr W' \rho - \ve}{\Tr W' \sigma},
\end{aligned}\end{equation}
where in the last line we made the choice of $V$ as projecting onto the first element of the two-element distribution $p_{\rho, M'}$. Optimising over all choices of $W'$ gives us $D^{\s{T}{sub}}_{\max,\,\MM} (\rho \| \sigma) \geq \wt{D}^\ve_{\max} (\rho \| \sigma) $, and so the two quantities must be equal.

The result for the variant $D^{\s{T}{n}}_{\max,\,\MM}$ with normalised smoothing follows since
\begin{equation}\begin{aligned}
    D^{\s{T}{n}}_{\max,\,\MM} (\rho \| \sigma) &= \sup_{M \in \MM}  D^{\s{T}{n}}_{\max}(p_{\rho, M} \| p_{\sigma, M}) \\
    &= \sup_{M \in \MM} \max \Big\{ 0,\, \Dmax{T}{s} (p_{\rho, M} \| p_{\sigma, M}) \Big\}\\
    &=  \max \Big\{ 0,\,  D^{\s{T}{s}}_{\max,\,\MM} (\rho \| \sigma) \Big\}
\end{aligned}\end{equation}
by Lemma~\ref{lem:classical_dmax_equivalence}.
\end{proof}


\subsection[\texorpdfstring%
{${D^\ve_H}$ as optimised ${\wt{D}^\ve_{\max}}$}%
{Hypothesis testing relative entropy as optimised modified max-relative entropy}]{$\boldsymbol{D^\ve_H}$ as optimised $\boldsymbol{\wt{D}^\ve_{\max}}$}

Although the connection between $D^{\ve}_{\max}$ and $D^{1-\ve}_H$ is now a well-known duality in one-shot quantum information theory, the precise bounds between the two are often rather loose --- they match in the asymptotic regime, and indeed also when the second-order asymptotics are considered~\cite{tomamichel_2013}, but large gaps are expected between them at the single-copy level (cf.~\eqref{eq:first_weakstrong}). 
Due to this, no exact equivalence has been obtained
between the two quantities.

Surprisingly, here we show that, as long as one considers the modified max-relative entropy variant $\wt{D}^\ve_{\max}$, the hypothesis testing relative entropy and the smooth max-relative entropy are in a precise sense equivalent to each other --- either of them can be obtained from the other.
This will allow us to obtain much tighter bounds between the two quantities than were previously known.

\begin{boxed}
\begin{theorem}\label{thm:equivalence_DH_Dtilde}
For all quantum states $\rho$ and $\sigma$ and for all {$\ve \in (0,1)$} it holds that
\begin{align}
  D^{1-\varepsilon}_H(\rho \| \sigma) &= \inf_{\delta \in [0, \varepsilon)}\left[ \widetilde{D}^{\delta}_{\max} (\rho \| \sigma)- \log (\varepsilon-\delta)\right],\label{eq:thm1}\\
  \widetilde{D}^{\varepsilon}_{\max}(\rho\|\sigma) &= \sup_{\delta \in (\varepsilon, 1]} \Big[ D^{1-\delta}_H (\rho \| \sigma) + \log (\delta - \varepsilon)\Big].\label{eq:thm2}
\end{align}
Eq.~\eqref{eq:thm1} holds also for $\ve=1$, while Eq.~\eqref{eq:thm2} for $\ve=0$. 
{Without loss of generality, the range of the optimisation can be restricted to $\delta \in (0,\ve)$ in~\eqref{eq:thm1} and $\delta \in (\ve,1)$ in~\eqref{eq:thm2}.}\\
If $\rho$ and $\sigma$ commute, $\wt{D}^{\delta}_{\max}$ can be replaced with $\Dmax{T}{sub}[\delta]$ in the above.
\end{theorem}
\end{boxed}

\begin{proof}
Take any $\ve \in (0,1]$. We first rewrite $D^{1-\ve}_H(\rho||\sigma)$ in a more convenient form using the change of variables $\frac{1}{\Tr M \sigma} \mapsto z$, $\frac{1}{\Tr M \sigma} M \mapsto M'$:
\begin{equation}\begin{aligned}\label{eq:dh_primal_mod}
        D^{1-\ve}_{H} (\rho \| \sigma) &= \log \sup \lset \frac{1}{\Tr M \sigma } \bar 0 \leq M \leq \id,\; \Tr M \rho \geq \ve \rset\\
        &= \log \sup_{z,M'} \lset z \bar 0 \leq M' \leq z \id,\; \Tr M' \rho \geq z \ve,\; \Tr M' \sigma = 1 \rset\\
        &= \log \sup_{z,M'} \lset z \bar 0 \leq M' \leq z \id,\; \Tr M' \rho \geq z \ve,\; \Tr M' \sigma \leq 1 \rset,
\end{aligned}\end{equation}
observing in the last line that the constraint $\Tr M' \sigma = 1$ can be relaxed to an inequality without loss of generality, as any feasible $M'$ with $\Tr M' \sigma < 1$ can be rescaled to satisfy $\Tr M' \sigma = 1$ while only increasing the feasible optimal value. 
One can then readily compute the Lagrange dual as
\begin{align}\label{eq:dh_dual}
    D^{1-\ve}_{H} (\rho \| \sigma) &= \log \inf_{h,y,P} \lset h \bar y \rho \leq h \sigma + P,\; P \geq 0,\; h,y \geq 0,\; 1 + \Tr P \leq y \ve \rset,
\end{align}
where equality follows by strong duality.\footnote{Strong duality is particularly easy to see for $\ve \in (0,1)$: a choice of $z \in (\ve,1)$ and $M' = \ve \id$ forms a strictly feasible solution pair to the primal problem~\eqref{eq:dh_primal_mod}, and so Slater's criterion applies (see e.g.~\cite[Sec.~5.9]{boyd_2004}). In the edge case of $\ve = 1$, where the value of the hypothesis testing relative entropy reduces to $D^0_H(\rho\|\sigma) = - \log \Tr \Pi_\rho \sigma$ with $\Pi_\rho$ denoting the projection onto the support of $\rho$, the primal problem in~\eqref{eq:dh_primal_mod} is no longer strictly feasible --- instead, we can argue the strong feasibility of the dual in~\eqref{eq:dh_dual}. To this end, assume that $\rho$ and $\sigma$ are not orthogonal, as otherwise the problem trivialises.
We would now like to find a $P > 0$ such that $y \rho < h \sigma + P$ and $1 + \Tr P < y $. Let us decompose the underlying Hilbert space $\mathcal{H}$ into $\mathcal{H} = \supp(\sigma) \oplus \operatorname{ker}(\sigma)$ and write 
$h \sigma + P - y \rho = \begin{pmatrix}h\sigma + \Pi_\sigma (P - y\rho) \Pi_\sigma & \Pi_\sigma (P - y\rho) \Pi_\sigma^\perp \\ \Pi_\sigma^\perp (P - y\rho) \Pi_\sigma & \Pi_\sigma^\perp (P - y\rho) \Pi_\sigma^\perp\end{pmatrix}$. For any fixed $P>0$ and $y>0$, taking the Schur complement with respect to the upper-left block tells us that, for all sufficiently large $h$, the matrix is positive definite if and only if $\Pi_\sigma^\perp ( P - y \rho) \Pi_\sigma^\perp > 0$ (understood as an operator acting only on $\ker(\sigma)$). Let us then choose $P \coloneqq \Pi_\sigma^\perp \rho \Pi_\sigma^\perp + \delta \id$ for some $\delta >0$. Then $1 + \Tr P = 1 + \delta \Tr \id + y \Tr \Pi_\sigma^\perp \rho$ which, for sufficiently large $y$, satisfies $1 + \Tr P < y$ since $ \Tr \Pi_\sigma^\perp \rho < 1$ by the assumption of non-orthogonality of $\rho$ and $\sigma$. The choice of $P$, $h$, and $y$ constructed in this way forms a strictly feasible solution to~\eqref{eq:dh_dual}, ensuring that strong duality holds by Slater's criterion.}
Substituting $\mu = \frac{1}{y}$, $Q = \mu P$, and $\lambda = h \mu$ we get
\begin{align}
    D^{1-\ve}_{H} (\rho \| \sigma) &= \log \inf_{h,\mu,P} \lsetr h \!\barr \rho \leq h \mu\, \sigma + \mu P,\; P \geq 0,\; h, \frac1\mu \geq 0,\; 1 + \Tr P \leq \frac1\mu \ve \rsetr\nonumber\\
    &= \log \inf_{h,\mu,Q} \lset h \bar \rho \leq h \mu\, \sigma + Q,\; Q \geq 0,\; h \geq 0,\; \mu > 0,\; \Tr Q \leq \ve - \mu \rset\nonumber\\
    &= \log \inf_{\mu \in (0,\ve]} \frac{1}\mu\, \inf_{\lambda,Q} \lsetr \lambda \barr \rho \leq \lambda \sigma + Q,\; Q \geq 0,\;  \Tr Q \leq \ve - \mu \rsetr\\
    &= \inf_{\mu \in (0,\ve]} \left[ \wt{D}^{\ve-\mu}_{\max}(\rho\|\sigma) - \log \mu \right],\nonumber
\end{align}
where in the last line we recalled the definition of~$\wt{D}^{\ve-\mu}_{\max}$ in~\eqref{eq:Dtilde_def}.
This is precisely the expression claimed in~\eqref{eq:thm1}. {Because of its definition as an infimum, the function $\delta \mapsto \widetilde{D}^\delta_{\max}$ is lower semi-continuous; since it is also easily verified to be non-increasing, it follows that it must be continuous from the right. This means that the endpoint $\mu=\ve$ can be excluded from the optimisation interval without affecting the optimal value. (For completeness, we prove the continuity properties of $\wt{D}^\delta_{\max}$ in full detail in Appendix~\ref{app:continuity}.)}

Let us now move on to the second equality of the theorem, namely~\eqref{eq:thm2}. Take any $\ve \in [0,1)$.
Using the dual form of $\wt{D}^\ve_{\max}$ in the last line of~\eqref{eq:Dtilde_dual}, we have
\begin{align}\label{eq:Dtilde_dual_tighter}
    \wt{D}^\ve_{\max} (\rho \| \sigma) &= \log \sup \lset \frac{\Tr W' \rho - \ve}{\Tr W' \sigma} \bar 0 \leq W' \leq \id \rset\nonumber\\
    &\texteq{(i)} \log \sup \lset \frac{\Tr W' \rho - \ve}{\Tr W' \sigma} \bar 0 \leq W' \leq \id,\; \Tr W' \rho > \ve \rset\nonumber\\
    &= \log \sup_{\delta \in (\ve,1]} \sup \lset \frac{\Tr W' \rho - \ve}{\Tr W' \sigma} \bar 0 \leq W' \leq \id,\; \Tr W' \rho = \delta \rset\\
    &=\log \sup_{\delta \in (\ve,1]} \sup \lset \frac{\delta - \ve}{\Tr W' \sigma} \bar 0 \leq W' \leq \id,\; \Tr W' \rho = \delta \rset\nonumber\\
    &\texteq{(ii)} \log \sup_{\delta \in (\ve,1]} \sup \lset \frac{\delta - \ve}{\Tr W' \sigma} \bar 0 \leq W' \leq \id,\; \Tr W' \rho \geq \delta \rset\nonumber\\
    &\texteq{(iii)} \sup_{\delta \in (\ve,1]} \left[ D^{1-\delta}_H(\rho \| \sigma) + \log (\delta - \ve) \right]\nonumber
\end{align}
where in (i) we observed that, since the optimal value of the supremum must be positive (consider that $W' = \id$ is feasible), we can restrict to operators $W'$ such that $\Tr W' \rho > \ve$ without loss of generality; in (ii), we used the fact that the condition $\Tr W' \rho = \delta$ can be relaxed to $\Tr W' \rho \geq \delta$ without loss of generality: for any $W''$ with $\Tr W'' \rho > \delta$ we can scale it down to a solution $W' \leq W''$ which has $\Tr W' \rho = \delta$ and for which the optimal value cannot be smaller; finally, (iii) is by definition of $D_H^{1-\delta}$.
{The fact that it suffices to optimise over $\delta \in (\ve,1)$ is argued as before: since the function $\delta \mapsto D^{1-\delta}_H$ is non-increasing and, because of its definition as a supremum, upper semi-continuous,\footnote{{Indeed, consider a sequence of values $(\e_n)_n$, with $\lim_n \e_n = \e$. Set $\lambda_n \coloneqq D_H^{\e_n} (\rho\|\sigma)$, and posit, up to extracting a subsequence, that $\lim_n \lambda_n = \limsup_n \lambda_n = \lambda$. By compactness, for all $n$ we can find a test $M_n$ such that $0\leq M_n \leq \id$, $\Tr M_n \rho \geq 1-\e_n$, and $\Tr M_n \sigma = \exp[-\lambda_n]$. Up to extracting yet another subsequence, we can further assume that $(M_n)_n$ has a limit, i.e.\ $\lim_n M_n = M$. Naturally, $M$ will also be a valid test, and $\Tr M\rho = \lim_n \Tr M_n \rho \geq \lim_n(1-\e_n) = 1-\e$. Hence, $D_H^\e(\rho\|\sigma) \geq - \log \Tr M\sigma = - \lim_n \log \Tr M_n \sigma = \limsup_n D_H^{\e_n}(\rho\|\sigma)$.}} it must be continuous from the left.}

The classical (commuting) case follows by Lemma~\ref{lem:classical_dmax_equivalence}.
\end{proof}

{
From the proof it is evident that this relation can be understood as a very direct consequence of  Lagrange duality, making it conceptually related to the quantum Neyman--Pearson lemma~\cite{HELSTROM}. 
Some related quantitative connections were observed before. This includes a dual expression of $D^\ve_H$ in terms of the hockey-stick divergence $\Tr (\rho - \lambda \sigma)_+$~\cite[Proposition~3.2]{audenaert_2012}, as well as the observation that~\eqref{eq:thm2} holds in the classical case~\cite[Proposition~2]{renes_2018}.
}

{%
We also note an interesting connection with one-shot results in information theory that appear in the form $D^{\ve -\delta}_H (\rho\|\sigma) - \log \frac{1}{\delta}$, notably in the recent one-shot achievability bounds of~\cite{cheng_2023-1}. By Theorem~\ref{thm:equivalence_DH_Dtilde}, optimising over $\delta$ allows such results to be understood instead in terms of $\wt{D}^{1-\ve}_{\max}$. More precisely, the smooth relative entropy bounds of~\cite{cheng_2023-1} concern bounding the achievable size of a code $M$ in terms of the demanded error probability $\ve$, and for instance~\cite[Proposition~2]{cheng_2023-1} can be understood as stating the bound
\begin{equation}\begin{aligned}\label{eq:haochung}
    \log M \geq \wt{D}^{1-\ve}_{\max} (\rho_{XB} \| \rho_X \otimes \rho_B),
\end{aligned}\end{equation}
with an analogous rephrasing applying to the other bounds of~\cite{cheng_2023-1}. Indeed, this could already be deduced from the proof approach of~\cite{cheng_2023-1} that directly employs the quantity $\Tr [\rho \wedge \lambda \sigma] = 1 - \Tr(\rho -\lambda \sigma)_+$, as was recently observed also in~\cite{gour_2025-1}. The rewriting in terms of $\wt{D}^{1-\ve}_{\max}$ removes the avoidable error terms $\delta$ in the bound and, through the results of this work, can help obtain insights into the properties of this achievability result as well as improved connections with other smooth relative entropies.
In the classical case, analogous achievability results were obtained in~\cite[Theorem~17]{polyanskiy_2010}, and similar quantitative connections were discussed in~\cite[\S 2.4]{polyanskiy_2010-1}.
}%

A noteworthy aspect of the classical case of the above result is that it gives a direct correspondence between the hypothesis testing relative entropy and the smooth max-relative entropy $\Dmax{T}{sub}$. Moreover, if $\ve \leq \frac12 \norm{p-q}{1}$, then we know from Lemma~\ref{lem:subnormalised_normalised_main} that we can equivalently use the normalised smooth max-relative entropy:
\begin{equation}\begin{aligned}
     D^{1-\varepsilon}_H(p \| q) &= \inf_{\delta \in [0, \varepsilon)}\left[ \Dmax{T}{n}[\delta] (p \| q)- \log (\varepsilon-\delta)\right],\\
     \Dmax{T}{n} (p \| q) &= \sup_{\delta \in (\varepsilon, 1]} \Big[ D^{1-\delta}_H ( p \| q) + \log (\delta - \varepsilon)\Big].
\end{aligned}\end{equation}
In the quantum case, the equivalence we previously showed in Proposition~\ref{prop:equivalence_measured} tells us that Theorem~\ref{thm:equivalence_DH_Dtilde} can be understood as a correspondence between the hypothesis testing relative entropy and the \emph{measured} smooth max-relative entropy. In Appendix~\ref{app:definitions} we also discuss other possible interpretations of $\wt{D}^\ve_{\max}$. However, in general $\wt{D}^\ve_{\max} (\rho \| \sigma) \neq \Dmax{T}{sub}(\rho\|\sigma)$, so the equivalence relation of Theorem~\ref{thm:equivalence_DH_Dtilde} does not extend to the standard smooth max-relative entropy $D_{\max}^{\ve}$ itself. Nevertheless, we will see in Section~\ref{sec:bounds} that the result will directly lead to the strengthening of multiple bounds involving $D^\ve_{\max}$.

Another immediate consequence of  Theorem~\ref{thm:equivalence_DH_Dtilde} is that, given two pairs of states $(\rho,\sigma)$ and $(\rho',\sigma')$, it holds that $D^\ve_H(\rho \| \sigma) \geq D^\ve_H(\rho' \| \sigma')$ $\forall \ve\in[0,1]$ if and only if $\wt{D}^\ve_{\max}(\rho \| \sigma) \geq \wt{D}^\ve_{\max}(\rho' \| \sigma') \; \forall \ve\in[0,1]$. (Note that the point $\ve=1$ is trivial in both cases, while for $\ve=0$ we have shown that the expressions in Theorem~\ref{thm:equivalence_DH_Dtilde} still apply.) This is related to the result of~\cite[Theorem~2]{buscemi_2017} and the classical~\cite[Theorem~3]{mu_2021}. {Statements of this kind are important in the context of statistical comparison of experiments in the sense of Blackwell~\cite{blackwell_1953} and its quantum generalisations~\cite{buscemi_2012-1,renes_2016,buscemi_2017}.}


\section{Improving a useful lemma}\label{sec:DR}

Although we have already shown a very close relation between $D^\ve_H$ and the modified quantity $\wt{D}^\ve_{\max}$, this is so far insufficient to tightly connect the hypothesis testing relative entropy with the smooth max-relative entropy $D^\ve_{\max}$ itself. The issue is that, although one can straightforwardly upper bound $\wt{D}^\ve_{\max}$ with $\Dmax{T}{sub}$ and other variants of $D^{\ve}_{\max}$ (see Eq.~\eqref{eq:dtilde_upper_bound_with_dmax}), the other direction is much less obvious: $\wt{D}^\ve_{\max}$ involves an optimisation over values of $\lambda$ such that $\rho \leq \lambda \sigma + Q$, where $Q$ is a positive semidefinite operator of sufficiently small trace, but how can 
this operator inequality yield a (possibly subnormalised) state $\rho' \approx_{\ve} \rho$ that can be used as a feasible solution for the smooth max-relative entropy?

A solution to this conundrum was first given in a fundamental lemma by Datta and Renner~\cite[Lemma~3]{datta_2009-1}. It found use in many results that studied the asymptotics of the max-relative entropy, e.g.~\cite{tomamichel_2009,brandao_2010,datta_2009,datta_2013-1,tomamichel_2016,lami_2024-2,lami_2024-1}. Here we introduce a modified proof approach, allowing us to give better estimates in particular for 
smoothing over normalised states. 

We first state the lemma in a rather general form; the specific applications to max-relative entropy bounds will be discussed in Section~\ref{sec:bounds}.

\begin{boxed}
\begin{theorem}[{(Tightened Datta--Renner lemma)}] \label{tightened_DR_lemma}
Let $\rho$ be a state, $A,Q\geq 0$ positive semi-definite with $\Tr Q \leq \ve < 1$, and assume that $\rho\leq A+Q$ holds. Then there exists a subnormalised state $\rho'$ such that 
\bb
\rho' \leq A,\, \qquad \frac{\rho'}{\Tr \rho'} \leq \frac{A}{1-\ve}\,;
\ee
moreover, $\rho'$ is close to $\rho$ in the following senses:
\begin{align}
F\big(\rho,\, \rho'\,\big) &\geq (1-\ve)^2\, , \label{tightened_DR_unnormalised_fidelity} \\
F\Big(\rho,\,\frac{\rho'}{\Tr \rho'}\Big) &\geq 1-\ve\, , \label{tightened_DR_fidelity} \\
\frac12 \left\| \rho - \frac{\rho'}{\Tr \rho'} \right\|_1 &\leq \sqrt{\ve}\, , \label{tightened_DR_trace_distance} \\
    \big\|\rho - \rho' \hspace{1pt}\big\|_+ &\leq \sqrt{\ve\left(1-\frac{3\ve}{4}\right)} + \frac\ve2 \leq \sqrt{\ve(2-\ve)} \, . \label{tightened_DR_gen_trace_distance}
\end{align}
\end{theorem}
\end{boxed}

Here, the fidelity bound with unnormalised states in~\eqref{tightened_DR_unnormalised_fidelity} is the same as the one found in~\cite{datta_2009-1,tomamichel_2016}, but the other bounds improve on known estimates (e.g.~\cite{datta_2009-1,brandao_2010}), and in particular on the statement of the lemma with normalised smoothing found in~\cite[Lemma~C.5]{brandao_2010}.

Before proving this result, let us discuss the simple but consequential difference between our approach and the one found in previous works. 
Since our aim is to use the matrix inequality $\rho \leq A + Q$ to construct a subnormalised state $\rho'$ such that $\rho' \leq A$, a natural idea is to define $T \coloneqq A^{1/2} (A + Q)^{-1/2}$ and take the ansatz $\rho' \coloneqq T \rho T^\dagger$, which is easily verified to satisfy the desired inequality. This is indeed the approach that the original proofs of the Datta--Renner lemma took~\cite{datta_2009-1,brandao_2010,tomamichel_2016}. But then how can we determine how close $\rho'$ is to $\rho$ in terms of 
distance measures such as the trace distance or fidelity? One may be tempted to use the celebrated gentle measurement {lemma~\cite{davies_1969,winter_1999},} which tells us that if $M$ is a POVM operator such that $\Tr M \rho$ is close to 1, then the subnormalised state $\sqrt{M} \rho \sqrt{M}$ is close to $\rho$. The issue, however, is that this lemma applies only to positive semidefinite operators, which the ansatz $T$ is not. This is indeed a difficulty that arises in the previous proof methods, and the bounds derived therein --- in particular, those for a normalised smoothing state --- are looser than what one would have obtained from applying the gentle measurement lemma. Is there, then, a way to instead pick a positive semidefinite operator in this approach?

Notice that the operator $G \coloneqq A^{1/2} U (A + Q)^{-1/2}$, where $U$ is some unitary, still satisfies $G \rho G^\dagger \leq A$. The question then becomes 
whether one can choose $U$ in a way that makes $G$ positive semidefinite. This is indeed always possible, and if $A > 0$, then such a solution is unique~\cite[Proposition~4.1.8]{bhatia_2007}: $U$ equals $(A^{-1/2}(A+Q)^{-1} A^{-1/2})^{1/2} A^{1/2} (A+Q)^{1/2}$, and $G$ becomes the \deff{geometric mean} of $A$ and $(A+Q)^{-1}$.

More generally, for $A,B \geq 0$ one defines the operator geometric mean $A \gm B$ as
\begin{equation}\begin{aligned}
    A \gm B \coloneqq A^{1/2} \left( A^{-1/2} B A^{-1/2} \right)^{1/2} A^{1/2}.
\end{aligned}\end{equation}
The properties of the geometric mean that will be relevant to us is that it is positive semidefinite, and that it is monotone non-increasing in either argument: in particular, $A \leq C$ implies that $A \gm B \leq C\gm B$~(see e.g.~\cite[Ch.~4.1]{bhatia_2007}).

It thus looks as if we are in a position to apply the gentle measurement lemma to $\rho' \coloneqq G \rho G$. Another issue transpires, however: although the known formulations of the gentle measurement lemma give tight estimates on the distance of the normalised state $\frac{G\rho G}{\Tr G^2 \rho}$ to $\rho$, it turns out that the previous bounds on the error of the \emph{sub}normalised state $G\rho G$ were not tight in trace distance. We thus first introduce the following improvement.

\begin{boxed}
\begin{lemma}[(A 
\emph{gentler} measurement lemma)] \label{tighter_gentle_measurement_lemma}
Let $M\in [0,\id]$ be a measurement operator and $\rho$ be an arbitrary state on the same system. If $\Tr M\rho \geq 1-\e$ for some $\e\in [0,1]$, then
\begin{align}
F\Big(\rho,\, \sqrt{M}\rho\sqrt{M}\Big) &\geq (1-\e)^2\, , \label{gentle_measurement_fidelity} \\
F\left(\rho,\, \frac{\sqrt{M}\rho \sqrt{M}}{\Tr M\rho}\right) &\geq 1-\e\, , \label{gentle_measurement_fidelity_norm} \\
\frac12 \left\|\rho - \frac{\sqrt{M}\rho \sqrt{M}}{\Tr M\rho} \right\|_1 &\leq \sqrt\e\, , \label{gentle_measurement_trace_distance} \\
\frac12 \left\|\rho - \sqrt{M}\rho\sqrt{M}\right\|_1 &\leq \begin{cases} \sqrt{\e\left(1-\frac{3\e}{4}\right)} & \text{\rm if $\e\leq 2/3$,} \\[1ex] 1/\sqrt3 & \text{\rm if $\e>2/3$} \end{cases} \leq \sqrt{\ve} \, ,\label{gentle_measurement_unnormalised_trace_distance} \\
 \left\|\rho - \sqrt{M}\rho\sqrt{M} \hspace{1pt}\right\|_+ &\leq \sqrt{\e\left(1-\frac{3\e}{4}\right)} + \frac\e2 \leq \sqrt{\ve(2-\ve)} \, . \label{gentle_measurement_gen_trace_distance}
\end{align}
All of the primary bounds are tight for all $\e\in [0,1]$.
\end{lemma}
\end{boxed}
Here, the bounds~\eqref{gentle_measurement_fidelity}--\eqref{gentle_measurement_trace_distance} were known (see e.g.~\cite[Lemma~9.4.1]{wilde_2017}), while~\eqref{gentle_measurement_unnormalised_trace_distance} and~\eqref{gentle_measurement_gen_trace_distance} improve on previous estimates~{\cite{davies_1969,winter_1999,ogawa_2007}. To streamline the discussion here, we} defer the proof to Appendix~\ref{app:gentle}.

We can now prove Theorem~\ref{tightened_DR_lemma} using the reasoning outlined above.

\begin{proof}[Proof of Theorem~\ref{tightened_DR_lemma}]
Up to projecting down onto the support of $A+Q$, we can assume without loss of generality that $A+Q>0$ is invertible. Define 
\bb
G\coloneqq A \gm \left((A+Q)^{-1}\right) = (A+Q)^{-1/2} \left( (A+Q)^{1/2} A (A+Q)^{1/2} \right)^{1/2} (A+Q)^{-1/2} .
\label{gm_proof_eq1}
\ee
Due to the monotonicity of the operator geometric mean we have that
\bb
0 \leq G \leq (A+Q) \,\# \left((A+Q)^{-1}\right) = \id\, .
\label{gm_proof_eq2}
\ee
Conjugating by $G$, from $\rho \leq A+Q$ we deduce that 
\bb
\rho' \coloneqq G\rho G \leq G (A+Q) G = A\, ,
\label{gm_proof_eq3}
\ee
as a simple calculation using the formula on the rightmost side of~\eqref{gm_proof_eq1} reveals. We now estimate
\bb
1 - \Tr \rho G^2 = \Tr \rho \left(\id-G^2\right) \leq \Tr (A+Q) \left(\id-G^2\right) = \Tr(A+Q) - \Tr G(A+Q)G = \Tr Q \leq \ve\, ,
\label{gm_proof_eq4}
\ee
where the first inequality holds due to the fact that $\id - G^2\geq 0$. Applying the gentler measurement lemma (Lemma~\ref{tighter_gentle_measurement_lemma}) gives the claimed bounds.
\end{proof}

An advantage of the proof approach using the operator geometric mean is that it can be easily adapted to more general scenarios. 
In particular, we can generalise it to a multi-partite setting where the smoothing is performed over the marginals over a global state, inspired by the notion of `simultaneous smoothing' considered in the context of the one-shot multiparty typicality conjecture by Drescher and Fawzi~\cite{drescher_2013} and encountered earlier in~\cite{anshu_2019}. We discuss this in detail in Section~\ref{sec:simultaneous}.


\section{Tightened bounds and relations}\label{sec:bounds}


\subsection[Relating \texorpdfstring{$\wt{D}^\ve_{\max}$ and $D^\ve_H$}{modified smooth max-relative entropy and hypothesis testing relative entropy}]{Relating $\boldsymbol{\wt{D}^\ve_{\max}}$ and $\boldsymbol{D^\ve_H}$}
\label{sec:bounds_dmax}

We are now ready to tackle the question of relating the two fundamental operational quantities, $D^\ve_{\max}$ and $D^\ve_H$.

The first key ingredient will be tight bounds between $D^\ve_H$ and the modified smooth max-relative entropy $\wt{D}_{\max}^\ve$, obtained from the precise connection between the two that we established in Theorem~\ref{thm:equivalence_DH_Dtilde}.

\begin{boxed}
\begin{lemma}\mbox{}
\label{lem:DH_Dtilde}
For all $\ve \in (0,1)$ and all $\mu \in (0,\ve]$, it holds that
\begin{equation}\begin{aligned} \label{eq:Dtildemax_vs_DH}
    \wt{D}^\ve_{\max} (\rho\|\sigma) + \log \frac{1}{\ve(1-\ve)} \leq D^{1-\ve}_H(\rho\|\sigma) \leq \wt{D}^{\ve-\mu}_{\max} (\rho\|\sigma) + \log\frac{1}{\mu}
\end{aligned}\end{equation}
\end{lemma}
\end{boxed}
The first inequality improves over the previously known bound of~\cite[Proposition~4.7]{datta_2015}, and indeed also over a stronger bound that was implicit in the proof of~\cite[Theorem~11]{datta_2013-1}. The second inequality was known~\cite{datta_2015}.

\begin{proof}
The upper bound is an immediate consequence of~\eqref{eq:thm1} in Theorem~\ref{thm:equivalence_DH_Dtilde}.

Let us now fix some $\delta \in (\ve,1]$. Using again the same equation~\eqref{eq:thm1} tells us that for any $\zeta \in (0,\delta]$ we have
\begin{equation}\begin{aligned}\label{eq:DH_upper}
    D^{1-\delta}_H (\rho \| \sigma) &\leq \wt{D}^{\delta-\zeta}_{\max}(\rho\|\sigma) + \log \frac1\zeta.
\end{aligned}\end{equation}
With the choice $\zeta = \delta - \ve + \mu$ for some fixed $\mu \in (0,\ve]$, this gives
\begin{equation}\begin{aligned}
     D^{1-\delta}_H (\rho \| \sigma) + \log\frac{\delta-\ve}{\ve(1-\ve)} &\leq \wt{D}^{\ve-\mu}_{\max}(\rho\|\sigma) +\log \frac{\delta-\ve}{\ve (1-\ve)(\delta - \ve + \mu)}\\
     &\leq \wt{D}^{\ve-\mu}_{\max}(\rho\|\sigma) +\log \frac{1-\ve}{\ve(1-\ve)(1 - \ve + \mu)}\\
     &= \wt{D}^{\ve-\mu}_{\max}(\rho\|\sigma) + \log \frac{1}{\ve(1 - \ve + \mu)},
 \end{aligned}\end{equation} 
 where in the second line we made use of the fact that $\frac{\delta-\ve}{\delta-\ve+\mu} = \left( 1 + \frac{\mu}{\delta - \ve} \right)^{-1}$, which is clearly monotonic in $\delta$. Consider now that
 \begin{equation}\begin{aligned}
     \ve(1-\ve+\mu) = \ve(1-\ve) + \ve \mu \geq \mu (1-\ve) + \ve \mu = \mu
 \end{aligned}\end{equation}
 since $\ve \geq \mu$ by assumption. Thus
 \begin{equation}\begin{aligned}
     D^{1-\delta}_H (\rho \| \sigma) + \log\frac{\delta-\ve}{\ve(1-\ve)} \leq \wt{D}^{\ve-\mu}_{\max}(\rho\|\sigma) + \log \frac1\mu.
 \end{aligned}\end{equation}
Taking the supremum over all $\delta \in (\ve,1]$ and infimum over all $\mu \in (0,\ve]$ yields
 \begin{equation}\begin{aligned}
     \wt{D}^\ve_{\max} (\rho\|\sigma) + \log \frac{1}{\ve(1-\ve)} \leq D^{1-\ve}_H(\rho\|\sigma)
 \end{aligned}\end{equation}
 by Theorem~\ref{thm:equivalence_DH_Dtilde}.
\end{proof}

A point of note here is that the bounds can be verified to be tight in many ways. The tightness of the error term $\log\frac{1}{\ve(1-\ve)}$ is particularly easy to see in the trivial case $\rho = \sigma$, where it holds that $D^{1-\ve}_H(\rho\|\rho)=\log\frac1\ve$ and $\wt{D}^\ve_{\max}(\rho\|\rho) = \log(1-\ve)$. In light of Theorem~\ref{thm:equivalence_DH_Dtilde}, we also see that the upper bound is as tight as possible, since it holds with equality by taking the infimum over $\mu$. The lower bound is additionally tight in an i.i.d.\ asymptotic sense at the level of exponents, as we will shortly see in Sec.~\ref{sec:bounds_renyi}.

{
\subsection[\texorpdfstring{Relating ${D^\ve_{\max}}$ and ${\wt{D}^\ve_{\max}}$}{Relating smooth and modified smooth max-relative entropy}]{Relating $\boldsymbol{D^\ve_{\max}}$ and $\boldsymbol{\wt{D}^\ve_{\max}}$}
}

Applying the improved Datta--Renner lemma immediately gives the following bounds.

\begin{boxed}
\begin{corollary}\label{cor:DR_dmax}
For all quantum states $\rho$ and $\sigma$ and for all $\ve \in [0,1)$, it holds that
\begin{align}
    \Dmax{T}{n}[\sqrt{\ve}](\rho\|\sigma) - \log\frac{1}{1-\ve} &\leq \wt{D}^{\ve}_{\max}(\rho\|\sigma),\\
    \Dmax{T}{s}[\sqrt{\ve\left(1-\frac{3\ve}{4}\right)} + \frac\ve2](\rho\|\sigma) &\leq \wt{D}^{\ve}_{\max}(\rho\|\sigma),\\
    \Dmax{P}{n}[\sqrt{\ve}](\rho\|\sigma) - \log\frac{1}{1-\ve} &\leq \wt{D}^{\ve}_{\max}(\rho\|\sigma),\label{eq:theoneweuse}\\
    \Dmax{P}{s}[\sqrt{\ve(2-\ve)}](\rho\|\sigma) &\leq \wt{D}^{\ve}_{\max}(\rho\|\sigma).
\end{align}
\end{corollary}
\end{boxed}

\begin{proof}
Recall that any feasible solution for $\wt{D}^{\ve}_{\max}(\rho\|\sigma)$ corresponds to an operator inequality $\rho \leq \lambda \sigma + Q$ for some $Q \geq 0$ with $\Tr Q \leq \ve$. The result is then a direct consequence of 
Theorem~\ref{tightened_DR_lemma} with the choice $A = \lambda \sigma$, as the result gives us precisely either normalised or subnormalised states that are feasible solutions for the different variants of $D^\ve_{\max}(\rho\|\sigma)$.
\end{proof}

To relate the modified max-relative entropy and the standard smoothed variants, we now need inequalities in the opposite direction. Indeed, one such relation follows immediately from the variational form of the trace distance in Eq.~\eqref{eq:trace_variational}. We state this here for completeness.

\begin{boxed}
\begin{lemma}\label{lem:dtilde_upper_bound_with_dmax}
For all $\ve \in (0,1)$, it holds that
\begin{equation}\begin{aligned}\label{eq:dtilde_upper_bound_with_dmax}
    \wt{D}^\ve_{\max} (\rho \| \sigma) &\leq \Dmax{T}{sub} (\rho \| \sigma) \leq \Dmax{P}{sub} (\rho \| \sigma),
\end{aligned}\end{equation}
and analogously for the quantities $\Dmax{T}{n}$ and $\Dmax{P}{n}$ with normalised smoothing.
\end{lemma}
\end{boxed}
\begin{proof}
Given any $\rho' \leq \lambda \sigma$, we have that $\rho \leq \lambda \sigma + (\rho - \rho')_+$, and by definition $\Tr (\rho - \rho')_+ = \norm{\rho-\rho'}{+}$. 
Thus any feasible solution for $\Dmax{T}{sub}$ yields a feasible solution for $\wt{D}^\ve_{\max}$ with the same feasible optimal value. 
Together with one of the Fuchs--van de Graaf inequalities~\eqref{eq:fvdg}, this gives the stated bounds.
\end{proof}

{%
This bound is tight for the trace distance, as can be seen from the classical case in Lemma~\ref{lem:classical_dmax_equivalence}. However, it can already be deduced from prior works such as~\cite[Proposition~13]{tomamichel_2013} that for the purified distance, tighter bounds can be obtained --- in particular, ones that feature $\sqrt\ve$ rather than $\ve$ in the upper bound. Here we improve on such bounds.

\begin{boxed}
\begin{lemma}\label{lem:purified_Dmax_lowerbound}
{%
Let $\rho,\sigma$ be two states. For all $\ve\in (0,1)$ and $\delta \in (0, 1-\ve)$, we have
\bb
\widetilde{D}_{\max}^{\ve + \delta} (\rho\|\sigma) \leq \Dmax{P}{s}[\sqrt\ve](\rho\|\sigma) + \log \frac{(\ve+\delta)(1-\ve-\delta)}{\delta}\, .
\label{lower_bound_Dmax}
\ee
As a consequence, for any $c \in (1, \frac{1}{\ve})$ it also holds that
\bb
\widetilde{D}_{\max}^{c \ve} (\rho\|\sigma) \leq \Dmax{P}{s}[\sqrt\ve](\rho\|\sigma) + \log \frac{c}{c-1}\, .
\label{eq:Dmax_lower_bound_multiplicative}
\ee
}%
\end{lemma}
\end{boxed}
\begin{proof}
Let $\rho'$ be a smoothing of $\rho$ with the property that
\bb
P(\rho,\rho') \leq \sqrt{\e}\, ,\qquad \rho'\leq \lambda \sigma\, ,\qquad \log \lambda = D_{\max}^{\sqrt\e,\,P}(\rho\|\sigma)\, .
\ee
By Uhlmann's theorem, this means that there are two purifications $\ket{\psi}$ of $\rho$ and $\ket{\psi'}$ of $\rho'$ (not necessarily normalised) such that
\bb
\Tr \psi \psi' = F(\rho,\rho') \geq 1 - \ve.
\ee
We would now like to estimate $\wt{D}^{\ve+\delta}_{\max}(\psi\|\psi')$. To this end, consider for some $k > 0$ the operator $\psi - k \psi'$. We will use the fact that any rank-two Hermitian operator $X$ satisfies
\begin{equation}\begin{aligned}\label{eq:rank_two}
    \norm{X}{1} = \sqrt{\Tr(X^2) + 2 \left|\det(X)\right|} = \sqrt{\Tr(X^2) + \left|\Tr(X^2) - \Tr(X)^2\right| }.
\end{aligned}\end{equation}
In our case, $X = \psi - k \psi'$ has one positive and one negative eigenvalue, which implies that  $\Tr(X^2) - \Tr(X)^2 \geq 0$. Thus
\begin{equation}\begin{aligned}\label{eq:pospart_pure_states}
    \norm{\psi - k \psi'}{1} &= \sqrt{2 \Tr((\psi-k\psi')^2) - \big(\!\Tr(\psi - k\psi')\big)^2}\\
    &= \sqrt{ \left(\Tr \psi + k \Tr \psi'\right)^2 - 4 k \Tr \psi \psi'}\\
    &\leq \sqrt{ \left(1 + k \right)^2 - 4 k (1-\ve)}\\
    &=\sqrt{ \left(1 - k\right)^2 + 4 k \ve}.
\end{aligned}\end{equation}
This gives
\begin{equation}\begin{aligned}
    \Tr(\psi - k \psi')_+ &= \frac12 \Tr (\psi - k \psi') + \frac12 \norm{\psi - k \psi'}{1}\\
    &\leq \frac12 \left( 1 - k + \sqrt{ (1-k)^2 + 4 k \ve } \right).
\end{aligned}\end{equation}
Setting the rightmost side to be equal to $\ve + \delta$ and solving for $k$, we obtain
\bb
k = \frac{(\ve+\delta)(1-\ve-\delta)}{\delta}.
\ee
By the data processing inequality, we then infer that
\bb\label{eq:Dtilde_dataproc}
\widetilde{D}_{\max}^{\ve+\delta}(\rho \| \rho') \leq \widetilde{D}_{\max}^{\ve+\delta}(\psi \| \psi') \leq \log \frac{(\ve+\delta)(1-\ve-\delta)}{\delta},
\ee
which implies that there exists an operator $Q \geq 0$ with $\Tr Q \leq \ve+\delta$ such that $\rho \leq \frac{(\ve+\delta)(1-\ve-\delta)}{\delta} \rho' + Q$.
Combined with the assumption that $\rho' \leq \lambda \sigma$ this gives
\bb
\widetilde{D}_{\max}^{\ve+\delta}(\rho \| \sigma) \leq \log \lambda \frac{(\ve+\delta)(1-\ve-\delta)}{\delta}
\ee
as claimed.

The second part of the Lemma follows by fixing $\delta = (c-1)\e$ and writing
\bb
\widetilde{D}_{\max}^{c\e}(\rho\|\sigma) \leq D_{\max}^{\sqrt\e,\,P,\,\leq}(\rho\|\sigma) + \log \frac{c\e (1-c\e)}{(c-1)\e} \leq D_{\max}^{\sqrt\e,\,P,\,\leq}(\rho\|\sigma) + \log \frac{c}{c-1}\, .
\ee
This concludes the proof.
\end{proof}

We observe that the approach of the proof readily gives an exact expression for the value of $\wt{D}^\ve_{\max}$ between any two pure states; this can be used to obtain an exact expression also for the pure-state hypothesis testing relative entropy $D^\ve_H$. Such an expression can be found for instance in~\cite[Ch.~IV, Eq.~(2.33)]{HELSTROM}, although we do not believe this to be well known in the literature. We derive and state the result here as it will be useful in establishing the tightness of our bounds.

\begin{boxed}
\begin{lemma}\label{lem:pure_state_formulas}
{%
Consider any two pure states $\psi_1 = \proj{\psi_1}$ and $\psi_2 = \proj{\psi_2}$, and let $f = \Tr \psi_1 \psi_2 =  \left|\braket{\psi_1|\psi_2}\right|^2$. Then
\begin{align}
    \wt{D}^\ve_{\max}(\psi_1\|\psi_2) &= \begin{dcases} %
        \log \frac{\ve(1-\ve)}{f - (1-\ve)} \quad \, & \text{if } \ve > 1 - f \label{eq:purestate_Dtilde}\\ 
        \infty & \text{otherwise}, 
    \end{dcases}\\
    D^\ve_{H}(\psi_1\|\psi_2) 
    &= \begin{dcases} %
        \log \frac{F_2\big(1- \ve,\, f\big)}{\big(f - \ve\big)^2} \quad\, & \text{if } \ve < f \\
        \infty & \text{otherwise}
    \end{dcases} \label{eq:purestate_DH}\\
    &= \begin{cases} %
      - \log \!\left[1 - F_2\!\left(\ve,\, f\right)\right] \quad\, & \text{if } \ve < f \\
        \infty & \text{otherwise},
    \end{cases} \label{eq:purestate_DH2}
\end{align}
where $F_2:[0,1]\times [0,1] \to [0,1]$ denotes the binary fidelity (Bhattacharyya coefficient), defined as
\begin{equation}\begin{aligned}\label{eq:bhattacharyya}
    F_2(p,q) \coloneqq \left(\sqrt{pq} + \sqrt{(1-p)(1-q)}\right)^2.
\end{aligned}\end{equation}
}%
\end{lemma}
\end{boxed}

\begin{proof}
Using Eq.~\eqref{eq:pospart_pure_states}, where we now write an equality instead of an inequality, we have
\begin{equation}\begin{aligned}
    \Tr( \psi_1 - \lambda \psi_2 )_+ = \frac12 \left( 1 - \lambda + \sqrt{(1+\lambda)^2 - 4 \lambda f} \right).
\end{aligned}\end{equation}
For the right-hand side to equal $\ve$, and hence for $\lambda$ to be a feasible solution for $\wt{D}^\ve_{\max}$, we need
\begin{equation}\begin{aligned}\label{eq:Dtilde_purestate}
     \lambda = \frac{\ve(1-\ve)}{f - (1-\ve)},
 \end{aligned}\end{equation}
 which is precisely the claimed expression. If $f \leq 1- \ve$, this is impossible to satisfy for any $\lambda \in (0,\infty)$, leading to a diverging optimal value.

For $D^\ve_H$, consider first that if $\ve \geq f$, then $M = \id - \psi_2$ satisfies $\Tr M \psi_1 \geq 1 -\ve$ and hence is a feasible measurement operator for the definition of the hypothesis testing relative entropy~\eqref{eq:dh_def}, giving $D^\ve_H(\psi_1\|\psi_2) \geq - \log \Tr M \psi_2  = \infty$. Let us then assume that $\ve < f$. In this case, take $M = 
\proj{\phi}$, where 
\begin{equation}\begin{aligned}
    \ket{\phi}= \sqrt{1-\ve} \ket{\psi_1} - \sqrt{\ve} \ket{\psi_1^\perp}, \qquad \ket{\psi_1^\perp} = \frac{(\id - \psi_1) \ket{\psi_2}}{\left\| (\id - \psi_1) \ket{\psi_2} \right\|} = \frac{(\id - \psi_1) \ket{\psi_2}}{\sqrt{1-f}}. 
\end{aligned}\end{equation}
This operator clearly satisfies $\Tr M \psi_1 = 1-\ve$, meaning that $D^\ve_H(\psi_1\|\psi_2) \geq - \log \Tr M \psi_2$, where 
\begin{equation}\begin{aligned}\label{eq:DH_pure_lower}
    \Tr M \psi_2 &= \left( \sqrt{1-\ve} \sqrt{f} - \sqrt{1-f} \sqrt{\ve} \right)^2\\
    &= \frac{\big(f - \ve\big)^2}{F_2\big(1- \ve,\, f\big)},
\end{aligned}\end{equation}
using that $(a-b)^2 = \frac{(a^2-b^2)^2}{(a+b)^2}$. 
The expression appearing in~\eqref{eq:purestate_DH2} is a rewriting of the above using the convenient identity
\begin{equation}\begin{aligned}\label{eq:binary_fidelity_identity}
\left(\sqrt{p (1-q)} - \sqrt{(1-p) q} \right)^2 = 1 - \left(\sqrt{p q} +\sqrt{(1-p) (1-q)} \right)^2.
\end{aligned}\end{equation}
For the opposite inequality, consider that Lemma~\ref{lem:DH_Dtilde} gives
\begin{equation}\begin{aligned}\label{eq:DH_pure_upper}
    D^\ve_{H}(\psi_1 \| \psi_2) &\leq \wt{D}^{1-\ve-\mu}_{\max} (\psi_1 \| \psi_2)  + \log \frac{1}{\mu}\\
    &= \log \frac{(1 - \ve - \mu)(\ve+\mu)}{(f-\ve-\mu)\mu}
\end{aligned}\end{equation}
for any $\mu \in (0,f-\ve)$, where in the second line we used the pure-state formula for $\wt{D}^{\ve-\mu}_{\max}$ from Eq.~\eqref{eq:Dtilde_purestate}. Choosing
\begin{equation}\begin{aligned}
    \mu = \frac{f-\ve}{1+ \sqrt{\frac{f(1-f)}{\ve(1-\ve)}}} = \frac{\ve(1-\ve) - \sqrt{\ve f(1-\ve)(1-f)}}{f - (1-\ve)}, 
\end{aligned}\end{equation}
some cumbersome yet straightforward algebra 
gives exactly
\begin{equation}\begin{aligned}
   \frac{(f-\ve-\mu)\mu}{(1 - \ve - \mu)(\ve+\mu)} 
   &= \left( \sqrt{1-\ve} \sqrt{f} - \sqrt{1-f} \sqrt{\ve} \right)^{2},
\end{aligned}\end{equation}
meaning that the upper bound of~\eqref{eq:DH_pure_upper} matches the lower bound of~\eqref{eq:DH_pure_lower}.
\end{proof}
}%


\subsection[\texorpdfstring{Improved bounds between $D^\ve_{\max}$ and $D^\ve_H$}{Improved bounds between max-relative entropy and hypothesis testing relative entropy}]{Improved bounds between $\boldsymbol{D^\ve_{\max}}$ and $\boldsymbol{D^\ve_H}$}

Putting our findings together, we obtain upper and lower bounds that directly connect $D^\ve_{\max}$ with $D^{1-\ve}_H$, recovering the known `weak/strong converse duality' 
between the two quantities and improving on the previously known quantitative bounds between them.

\begin{boxed}[filled]
\begin{theorem}[{({Tight} weak/strong converse duality between $D^\ve_{\max}$ and $D^{1-\ve}_H$)}] \label{cor:wsc}
For all quantum states $\rho$ and $\sigma$, all $\ve \in (0,1)$, and all $\mu\in (0,\ve]$, it holds that
\begin{align}
\Dmax{T}{n}[\sqrt{\ve}](\rho\|\sigma)  + \log \frac1\ve \leq D^{1-\ve}_H(\rho\|\sigma) &\leq \Dmax{T}{n}[\ve-\mu](\rho\|\sigma) + \log \frac{1}\mu\, , \label{eq:wsc_trace}\\
\Dmax{P}{n}[\sqrt{\ve}](\rho\|\sigma)  + \log\frac1\ve \leq D^{1-\ve}_H(\rho\|\sigma) &\leq 
{ \Dmax{P}{n}[\sqrt{\ve-\mu}](\rho\|\sigma) + \log \frac{F_2(1-\ve, \ve-\mu)}{\mu^2}} \label{eq:wsc_purified}\\
&\leq  { \Dmax{P}{n}[\sqrt{\ve-\mu}](\rho\|\sigma) + \log \frac{1}{\mu^2} } \,, \label{eq:wsc_purified2}
\end{align}
{where we recall that $F_2(1-\ve, \ve-\mu) = \left( \sqrt{(1-\ve) (\ve-\mu)} + \sqrt{\ve(1-\ve+\mu)}\right)^2$.}

For classical systems or commuting quantum states, a stronger trace distance bound holds:
\begin{equation}\begin{aligned}\label{eq:wsc_classical}
    \Dmax{T}{n} (p \| q) + \log\frac{1}{\ve} \leq D^{1-\ve}_H(p \| q) \leq \Dmax{T}{n}[\ve-\mu](p\|q) + \log \frac{1}\mu.
\end{aligned}\end{equation}

{All of the primary bounds of this theorem are tight: namely, for all $0<\mu\leq \e<1$, we can find two states saturating any of the inequalities in~\eqref{eq:wsc_trace},~\eqref{eq:wsc_purified}, and~\eqref{eq:wsc_classical}. 
The bounds also hold when normalised smoothing is replaced with subnormalised smoothing. }
\end{theorem}
\end{boxed}

\begin{proof}
The lower bounds on the hypothesis testing relative entropy in Eqs.~\eqref{eq:wsc_trace}--\eqref{eq:wsc_purified} follow directly by combining Lemma~\ref{lem:DH_Dtilde} and Corollary~\ref{cor:DR_dmax}. The upper bound in~\eqref{eq:wsc_trace} is from Lemma~\ref{lem:dtilde_upper_bound_with_dmax}.

{%
The upper bound in Eq.~\eqref{eq:wsc_purified} requires more elaboration. Combining Lemma~\ref{lem:DH_Dtilde} with Lemma~\ref{lem:purified_Dmax_lowerbound} gives, for any $\delta \in (0,\mu)$, that
\bb
D_H^{1-\e}(\rho\|\sigma) &\leq \widetilde{D}_{\max}^{\e+\delta-\mu}(\rho\|\sigma) + \log\frac{1}{\mu-\delta} \\
&\leq \Dmax{P}{s}[\sqrt{\e-\mu}](\rho\|\sigma) + \log\frac{(\e+\delta-\mu)(1-\e-\delta+\mu)}{\delta(\mu-\delta)}\, .
\ee
For the choice of $\delta = \mu \left( 1 + \sqrt{\frac{\ve(1-\ve)}{(\ve-\mu)(1-\ve+\mu)}} \right)^{-1}$, 
as in the proof of Lemma~\ref{lem:pure_state_formulas} we get exactly~\eqref{eq:wsc_purified}, which can be verified to be minimal.
}%

For the classical case, we can use Lemma~\ref{lem:subnormalised_normalised_main} as well as the equivalence between $\Dmax{T}{n}(p \| q)$ and $\wt{D}^\ve_{\max}(p\|q)$ in Lemma~\ref{lem:classical_dmax_equivalence} to~get
\begin{equation}\begin{aligned}
    \Dmax{T}{n}(p\|q)  +\frac1\ve = \max \left\{ \wt{D}_{\max}^\ve (p\|q) + \log \frac1\ve,\; \log \frac1\ve \right\} \leq D^{1-\ve}_H(p\|q),
\end{aligned}\end{equation}
where the inequality follows from Lemma~\ref{lem:DH_Dtilde} as well as the fact that $D^{1-\ve}_H(\rho\|\sigma) \geq \log \frac{1}{\ve}$ for all states, easily verified by choosing $M = \ve \id$ in the definition~\eqref{eq:dh_def}.
\end{proof}

{Almost all of the inequalities in Theorem~\ref{cor:wsc} improve over state-of-the-art bounds, often significantly so.} The first inequalities in Eqs.~\eqref{eq:wsc_trace} and~\eqref{eq:wsc_purified} are a major improvement over the best known bound of~\cite[Theorem~4]{anshu_2019}: the left-hand side of our bound is larger by an additive term of $\log\frac{1}{\ve(1-\ve)}$.
This result improves also on bounds stated in~\cite[Proposition~4.1]{dupuis_2012} and~\cite[Theorem~11]{datta_2013-1}, which had tighter error terms but worse smoothing terms than the bound of~\cite{anshu_2019}.
{The second inequality in~\eqref{eq:wsc_trace} recovers a result of~\cite[Theorem~11]{datta_2013-1}. Our upper bound for the purified distance in~\eqref{eq:wsc_purified} improves  previously known results, the strongest of which was the upper bound of~\cite[Theorem~4]{anshu_2019}.}

{Once again, the bounds are 
the strongest possible.} Even in the simplest case $\rho=\sigma$, previous results did not give tight bounds; in contrast, our bounds are, to the best of our knowledge, the first statement of the weak/strong converse duality between $D^\ve_H$ and $D^\ve_{\max}$ that gives tight error terms in this sense.
{This follows since $\Dmax{T}{n}[\sqrt\ve](\rho\|\rho) = \Dmax{P}{n}[\sqrt\ve](\rho\|\rho) = 0$ while $D^{1-\ve}_H(\rho\|\rho) = \log \frac1\ve$, which immediately yields a saturation of the lower bounds on $D^{1-\ve}_H$ in Eqs.~\eqref{eq:wsc_trace}--\eqref{eq:wsc_purified}.} 
The lower bounds also give a tight constraint on the asymptotic error exponent of~$\Dmax{P}{s}$~\cite{li_2023} and on the exponent of $\Dmax{T}{s}$ for classical systems, which we will discuss in more detail in the next section.

{%
The upper bounds 
on the hypothesis testing relative entropy in~\eqref{eq:wsc_trace} and~\eqref{eq:wsc_classical} are saturated, for example, by the classical distributions $p = (\ve, 1\!-\!\ve)$ and $q = (\mu, 1\!-\!\mu)$, since one easily verifies that $\|p - q\|_+ = \epsilon\!-\!\mu$ and so $\Dmax{T}{n}[\epsilon-\mu](p\|q) = 0$, while the test $M = (1,0)$ satisfies $\Tr M p = \epsilon$ and $\Tr M q = \mu$, meaning that $D^{1-\epsilon}_H(p\|q) \geq - \log \mu$. 
For the upper bound in~\eqref{eq:wsc_purified}, consider two pure states $\psi = \proj{\psi}$ and $\phi = \proj{\phi}$ such that $P(\psi,\phi) = \sqrt{\vphantom{k}\ve-\mu}$ and hence $\Dmax{P}{n}[\sqrt{\ve-\mu}] (\psi\|\phi) = 0$. By Lemma~\ref{lem:pure_state_formulas}, $D^{1-\ve}_H(\psi\|\phi) = - \log (1- F_2(\ve, \ve-\mu))$, which exactly matches the rightmost side of the bound in~\eqref{eq:wsc_purified}.

The upper bound in~\eqref{eq:wsc_classical} is also tight for any $p$, $q$, and $\ve \leq \norm{p-q}{+}$, in the following sense: 
by Theorem~\ref{thm:equivalence_DH_Dtilde} and Lemma~\ref{lem:subnormalised_normalised_main}, 
minimising the rightmost side of~\eqref{eq:wsc_trace} or~\eqref{eq:wsc_classical} over $\mu$ gives exact equality in the bound. This means that the upper bound can be saturated for all classical distributions with the optimal choice of~$\mu$.

We observe that the bounds 
in~\eqref{eq:wsc_trace} have a mismatch 
in the functional dependence of the smoothing parameter: the lower bound on $D^{1-\ve}_H$ features smooth max-relative entropy with smoothing of order $\sqrt{\ve}$ (effectively due to the use of the gentle measurement lemma in Theorem~\ref{tightened_DR_lemma}), while the smoothing parameter in the upper bounds is of order $\ve$. 
This contrasts with the bounds for purified distance in~\eqref{eq:wsc_purified}, which both share the same order of $\sqrt\ve$ --- this is a crucial property that allowed e.g.\ 
the computation of the second-order expansion of the max-relative entropy for the purified distance~\cite{tomamichel_2013} and the evaluation of the error and strong converse exponents for $\Dmax{P}{sub}$~\cite{li_2023,li_2024-1}.
It is clear that the trace distance smoothing exhibits a different behaviour than the purified distance, as indicated by the fact that the classical case of~\eqref{eq:wsc_classical} features matching scaling of order $\ve$ rather than $\sqrt\ve$. Indeed, because of this classical case, we know that the term $\ve$ in the upper bound cannot be improved to $\sqrt\ve$. However, one could still ask: is it possible that we could instead improve the lower bound for the trace distance, e.g.\ by establishing that $\Dmax{T}{n} + \log \frac1\ve \leq D^{1-\ve}_H$ holds for all quantum states? This is in fact impossible, as we now argue. Consider two pure states $\psi$ and $\phi$ with trace distance $\norm{\psi-\phi}{+} = \sqrt{1-\left|\braket{\psi|\phi}\right|^2} = \sqrt\ve$. From Lemma~\ref{lem:pure_state_formulas} we know that $D^{\delta}_{H}(\psi \| \phi)$ is infinite iff $\delta \geq 1-\ve$, while it is not difficult to verify that $\Dmax{T}{n}[\delta](\psi\|\phi)$ is infinite iff $\delta < \sqrt\ve$. Hence, any relation of the form
\begin{equation}\begin{aligned}
    \Dmax{T}{n}[\delta](\psi\|\phi) + g(\delta) \leq D^{f(\delta)}_H(\psi\|\phi)
\end{aligned}\end{equation}
must be such that $f(\delta) \geq 1-\ve$ for all $\delta < \sqrt\ve$. The best case scenario is therefore $f(\sqrt\ve) = 1-\ve$, which precisely corresponds to the bound in~\eqref{eq:wsc_trace}. The fact that the choice of $g(\sqrt\ve) = \log \frac{1}{\ve}$ is optimal in general can be verified by considering the trivial case $\psi = \phi$, where $\Dmax{T}{n}(\psi\|\psi) = 0$ but $D^{1-\ve}_H(\psi\|\psi) = \log\frac{1}{\ve}$. Our bounds are therefore tight also in this sense, and the smoothing of $\Dmax{T}{n}$ exhibits an inherent difference in the scaling of the error with respect to $D^{1-\ve}_H$ between pure states and classical states.
}%

We additionally note that, in light of Lemma~\ref{lem:subnormalised_normalised_main}, potentially tighter one-shot restrictions can be obtained by using the subnormalised smoothing variants of the max-relative entropy. However, the scaling and asymptotic behaviour of these bounds is essentially the same as the ones given in Theorem~\ref{cor:wsc}. We state the bounds here for completeness.
\begin{boxed}
\begin{corollary}
\label{cor:wsc_sb}
    For all quantum states $\rho$ and $\sigma$ and all $\ve \in (0,1)$, it holds that
\begin{align}
\Dmax{T}{sub}[\sqrt{\ve\left(1-\frac{3\ve}{4}\right)} + \frac\ve2](\rho\|\sigma) + \log\frac{1}{\ve(1-\ve)} \leq\; & D^{1-\ve}_H(\rho\|\sigma), \label{eq:wsc_trace_sub}\\
\Dmax{P}{sub}[\sqrt{\ve(2-\ve)}](\rho\|\sigma) + \log \frac{1}{\ve(1-\ve)} \leq\;& D^{1-\ve}_H(\rho\|\sigma).\label{eq:wsc_purified_sub}
\end{align}
For classical systems or commuting quantum states, we also have
\begin{equation}\begin{aligned}\label{eq:wsc_classical_sub}
    \Dmax{T}{sub} (p \| q) + \log\frac{1}{\ve(1-\ve)} \leq D^{1-\ve}_H(p \| q).
\end{aligned}\end{equation}
\end{corollary}
\end{boxed}
These bounds are a direct consequence of Lemma~\ref{lem:DH_Dtilde} and Corollary~\ref{cor:DR_dmax}, with the classical result using Lemma~\ref{lem:classical_dmax_equivalence}. They improve on previously known inequalities that used subnormalised smoothing, e.g.~\cite{dupuis_2012}.

{We remark that an expression of the same exact form as our upper bound on the rightmost side of~\eqref{eq:wsc_purified} previously appeared in~\cite[Lemma~III.8]{ramakrishnan_2023} as a tight upper bound on the smooth min-relative entropy (sandwiched R\'enyi divergence of order $1/2$). It is unclear to us if either of the bounds implies the other or if they are incomparable in general.}


\section{Consequences and other inequalities}

\subsection{Inequalities with R\'enyi relative entropies}\label{sec:bounds_renyi}

The Petz--R\'enyi relative entropies $D_\alpha$~\cite{petz_1986} and the sandwiched R\'enyi relative entropies $\wt{D}_\alpha$~\cite{muller-lennert_2013,wilde_2014} are defined, respectively, as
\begin{align}
D_\alpha (\rho \| \sigma) &\coloneqq \frac{1}{\alpha-1} \log \Tr \left( \rho^\alpha \sigma^{1-\alpha}\right), \label{Petz_Renyi} \\
\wt{D}_\alpha (\rho \| \sigma) &\coloneqq \frac{1}{\alpha-1} \log \Tr \left(\sigma^{\frac{1-\alpha}{2\alpha}} \rho \sigma^{\frac{1-\alpha}{2\alpha}}\right)^\alpha.
\label{sandwiched_Renyi}
\end{align}
{For $\alpha \in (0,1)$, it is sufficient that $\rho$ and $\sigma$ not be orthogonal to ensure that $D_\alpha(\rho\|\sigma)$ and $\wt{D}_\alpha(\rho\|\sigma)$ are finite; for $\alpha > 1$, the quantities are set to infinity whenever $\supp(\rho) \not\subseteq \supp(\sigma)$.}
Both of the relative entropies are additive under tensor products: $D_\alpha (\rho^{\otimes n} \| \sigma^{\otimes n}) = n D_\alpha(\rho\|\sigma)$ and $\wt{D}_\alpha (\rho^{\otimes n} \| \sigma^{\otimes n}) = n \wt{D}_\alpha(\rho\|\sigma)$.

We will also employ the measured variant of the quantities, given, in analogy with~\eqref{eq-17}, by
\bb
D_{\alpha,\,\MM}(\rho\|\sigma) \coloneqq \sup \lset D_{\alpha}\! \left( p_{\rho,M} \| p_{\sigma,M} \right) \bar M = (M_i)_{i=1}^n \in \MM,\; n \in \NN_+ \rset,
\label{measured_Renyi}
\ee
where $p_{\rho,M}(i) = \Tr M_i \rho$. Note that the two different quantum definitions~\eqref{Petz_Renyi} and~\eqref{sandwiched_Renyi} lead to the same notion of measured R\'enyi relative entropy~\eqref{measured_Renyi}, because they coincide for all classical (commuting) states. 

An important property of $\wt{D}_\alpha$ is that it is asymptotically attained by measurements (see~\cite[Corollary~III.8]{mosonyi_2015} and~\cite[Corollary~4]{hayashi_2016-1}):
\begin{equation}\begin{aligned}\label{eq:sandwich_attained}
    \wt{D}_\alpha (\rho \| \sigma) = \lim_{n\to\infty} \frac1n D_{\alpha,\MM}\left(\left.\rho^{\otimes n} \right\| \sigma^{\otimes n}\right).
\end{aligned}\end{equation}

Several bounds were given in the literature that connect smooth entropies with R\'enyi $\alpha$ divergences. Here we discuss how they can be improved using the relations established in this work.

First, we obtain upper bounds.
\begin{boxed}
\begin{corollary}\label{cor:var-mrel}
For all $\e\in (0,1)$ and all $\alpha > 1$, it holds that
\begin{align}
    \wt{D}^\ve_{\max}(\rho \| \sigma) + \log\frac{1}{1-\ve} &\leq D_{\alpha,\MM} (\rho \| \sigma) + \frac{1}{\alpha-1} \log\frac{1}{\ve} \leq \wt{D}_{\alpha} (\rho \| \sigma) + \frac{1}{\alpha-1} \log\frac{1}{\ve}.\label{eq:renyi_dtilde}
\end{align}
As a result,
\begin{equation}\begin{aligned}
\Dmax{T}{n}(\rho \| \sigma) \leq \Dmax{P}{n}(\rho \| \sigma) &\leq D_{\alpha,\MM} (\rho \| \sigma) + \frac{1}{\alpha-1} \log \frac{1}{\ve^2}\\ &\leq \wt{D}_{\alpha} (\rho \| \sigma) + \frac{1}{\alpha-1} \log \frac{1}{\ve^2} .\label{eq:renyi_dmax}
\end{aligned}\end{equation}
\end{corollary}
\end{boxed}
The bound in~\eqref{eq:renyi_dmax} improves over previously known bounds in~\cite[Theorem~3]{anshu_2019} (see also~\cite[Proposition~6]{wang_2019}) and in~\cite[Proposition~6.22]{tomamichel_2016}, losing a superfluous additive factor of $\log\frac{1}{1-\ve^2}$ in the former and tightening the smoothing term in the latter. We remark here that 
some of the previous bounds were stated in terms of the looser sandwiched R\'enyi relative entropies, but it is clear from their proofs that they apply to $D_{\alpha,\MM}$ too.

\begin{proof}[Proof of Corollary~\ref{cor:var-mrel}]
By Lemma~\ref{lem:DH_Dtilde} we have that
\begin{equation}\begin{aligned}\label{eq:DH_upper_Dalpha}
    \wt{D}^\ve_{\max}(\rho \| \sigma) + \log \frac{1}{\ve(1-\ve)} \leq D^{1-\ve}_H (\rho\|\sigma)\leq D_{\alpha,\MM} (\rho \| \sigma) + \frac{\alpha}{\alpha-1} \log \frac{1}{\ve}
\end{aligned}\end{equation}
for all $\alpha > 1$, where the second inequality is a standard argument based on the data processing of the R\'enyi divergences (see e.g.~\cite[Lemma~IV.7]{mosonyi_2015}). 
Using the Datta--Renner lemma {(in particular, Eq.~\eqref{eq:theoneweuse} in Corollary~\ref{cor:DR_dmax}) then gives the bound for the smooth max-relative entropy, noting that in~\eqref{eq:renyi_dmax} we changed the smoothing parameter from $\ve$ to $\ve^2$. }
The fact that $D_{\alpha,\MM} (\rho \| \sigma) \leq \wt{D}_\alpha (\rho \| \sigma)$ is a consequence of the data processing inequality for the sandwiched R\'enyi divergence for all $\alpha>1$~\cite{frank_2013,beigi_2013}.
\end{proof}

To investigate the tightness of the bounds, {we turn to a large-deviation--style analysis in terms of error exponents. Here we will use this latter term in an abstract, rather than operational, sense: we do not study any particular information-theoretic task, but rather assume that the smoothing parameter $\ve$ of $\wt{D}^\ve_{\max}$ can be understood as the error of some relevant problem. We can then look at the `error exponent'}
of the quantity $\wt{D}^\ve_{\max}$, that is, 
the largest exponent $E$ such that $\wt{D}^{\exp\left({-nE + o(n)}\right)}_{\max}(\rho^{\otimes n} \| \sigma^{\otimes n}) = nR$ for some fixed rate $R > 0$. 

Plugging $\ve = \exp\left({-nE + o(n)}\right)$ into~\eqref{eq:renyi_dtilde} and dividing by $n$, we have in the limit $n\to\infty$ that
\begin{equation}\begin{aligned}\label{eq:renyi_bound_upper_exponent}
    E &\geq \sup_{\alpha > 1}\, (\alpha-1) \left( R - \lim_{n\to\infty}\frac1n D_{\alpha,\MM} (\rho^{\otimes n} \| \sigma^{\otimes n}) \right)\\
    &= \sup_{\alpha > 1}\, (\alpha-1) \left( R - \wt{D}_\alpha(\rho \| \sigma) \right),
\end{aligned}\end{equation}
where we recalled~\eqref{eq:sandwich_attained}. 
This asymptotic bound is in fact known to be tight: this was established in~\cite[Theorem~IV.4]{mosonyi_2015} as a key step in the derivation of the strong converse exponent of quantum hypothesis testing.

For the smooth max-relative entropy, \eqref{eq:renyi_dmax}~gives an asymptotically tight bound on the error exponent of $\Dmax{P}{sub}$ (and hence, by Lemma~\ref{lem:subnormalised_normalised_main}, also of $\Dmax{P}{n}$). {To see this, notice that combining Corollary~\ref{cor:DR_dmax} with Lemma~\ref{lem:purified_Dmax_lowerbound} says that for any constant $c \in (1, \frac1\ve)$ we have
\begin{equation}\begin{aligned}
    \Dmax{P}{s}[\sqrt{2c\ve}] (\rho\|\sigma) \leq \wt{D}^{c \ve}_{\max}(\rho\|\sigma) \leq \Dmax{P}{s}[\sqrt{\ve}](\rho\|\sigma) + \log \frac{c}{c-1}.
\end{aligned}\end{equation}
As constant factors do not affect the asymptotic behaviour when the error vanishes exponentially fast, this shows that the error exponent of $\Dmax{P}{s}$ is exactly half that of $\wt{D}^\ve_{\max}$, i.e.\ half of~\eqref{eq:renyi_bound_upper_exponent}, as previously shown in~\cite{li_2023}.
} 
From Lemma~\ref{lem:classical_dmax_equivalence} and the above discussion, we also know that the bound of~\eqref{eq:renyi_bound_upper_exponent} gives exactly the error exponent of $\Dmax{T}{sub}$ when $\rho$ and $\sigma$ commute.

We can also give a lower bound.
\begin{boxed}
\begin{corollary}\label{cor:renyi_lower}
For all states  $\rho,\sigma$, all $\ve \in (0,1)$, and all $\alpha \in (0,1)$, it holds that
\begin{align}\label{eq:Dmax_lower_alpha}
    \Dmax{P}{sub}(\rho\|\sigma) \geq \Dmax{T}{sub} (\rho\|\sigma) \geq \wt{D}^\ve_{\max} (\rho\|\sigma) \geq D_\alpha(\rho \| \sigma) - \frac{1}{1-\alpha} \log \frac{1}{1-\ve}.
    \end{align}
    Consequently,
    \begin{equation}\begin{aligned}\label{eq:DH_lower_alpha}
    D^\ve_H(\rho\|\sigma) \geq D_\alpha(\rho \| \sigma) - \frac{\alpha}{1-\alpha} \log \frac{1}{\ve} + \log \frac{1}{1-\ve}.
\end{aligned}\end{equation}
\end{corollary}
\end{boxed}

As a bound on the smooth max-relative entropy, Eq.~\eqref{eq:Dmax_lower_alpha} improves on the bound given in~\cite[Proposition~4]{wang_2019}.  The bound on the hypothesis testing relative entropy in~\eqref{eq:DH_lower_alpha} is in general incomparable with the best previously known bound of~\cite[Proposition~3.2]{audenaert_2012}, which lacks the final $\log\frac{1}{1-\ve}$ term but instead features an extra $\alpha$-dependent term. {The additional $\log\frac{1}{1-\ve}$ term can, for instance, improve results that employ $D^\ve_H$ as an intermediate quantity for connections with R\'enyi divergences, such as a bound on the smooth divergence $\wt{D}_{1/2}$ in~\cite[Remark~4]{nuradha_2024-1}.}

\begin{proof}[Proof of Corollary~\ref{cor:renyi_lower}]
{When $\rho \perp \sigma$, all of the terms diverge, and the statement is trivial; let us thus assume otherwise.}
The first two inequalities are then immediate from the definitions (see Lemma~\ref{lem:dtilde_upper_bound_with_dmax}). The last one 
is an application of the inequality of Audenaert et al.~\cite{audenaert_2007}, which states that $\Tr (A-B)_+ \geq \Tr A - \Tr A^{\alpha} B^{1-\alpha}$ for all operators $A,B\geq 0$ and all $\alpha \in (0,1)$. Specifically,
\begin{equation}\begin{aligned}
    \wt{D}^\ve_{\max} (\rho\|\sigma) &= \inf \lset \log \lambda \bar \Tr (\rho - \lambda \sigma)_+ \leq \ve \rset\\
    &\geq \inf \lset \log \lambda \bar 1 - \lambda^{1-\alpha} \Tr \rho^{\alpha} \sigma^{1-\alpha} \leq \ve \rset\\
    &= \inf \lsetr \log \lambda \barr \log \lambda\! \geq \frac{1}{1-\alpha}\log(1-\ve) - \frac{1}{1-\alpha} \log \Tr \rho^{\alpha} \sigma^{1-\alpha} \rsetr\\
    &= \frac{1}{\alpha-1} \log \Tr \rho^{\alpha} \sigma^{1-\alpha} - \frac{1}{1-\alpha} \log \frac{1}{1-\ve}.
\end{aligned}\end{equation}
The bound on $D^\ve_H$ in~\eqref{eq:DH_lower_alpha} then follows by combining this result with Lemma~\ref{lem:DH_Dtilde}.
\end{proof}

Once again, to study the tightness of the bound, let us look at exponents --- now the strong converse exponent of $\wt{D}^\ve_{\max}$, namely the least $E_{\rm sc}$ such that $\wt{D}^{1- \exp\left({-nE_{\rm sc} + o(n)}\right)}_{\max}(\rho^{\otimes n} \| \sigma^{\otimes n}) = nR$ for some fixed $R>0$. Corollary~\ref{cor:renyi_lower} gives a lower bound on this exponent as
\begin{equation}\begin{aligned}\label{eq:renyi_lower_bound_exponent}
    E_{\rm sc} \geq \sup_{\alpha \in (0,1)} (\alpha-1)\left( R - D_\alpha (\rho \| \sigma) \right).
\end{aligned}\end{equation}
This is indeed tight, {as can be shown by contradiction. 
Take $\wt{D}^{1- \exp\left(-nE_{\rm sc} + o(n)\right)}_{\max}(\rho^{\otimes n} \| \sigma^{\otimes n}) = nR $ and assume 
that the inequality in~\eqref{eq:renyi_lower_bound_exponent} is strict, namely that $ E_{\rm sc} \geq (\alpha-1)( R - D_\alpha (\rho \| \sigma)) + \delta$ for all $\alpha \in(0,1)$ and some fixed $\delta > 0$. Relaxing Lemma~\ref{lem:DH_Dtilde} gives us the inequality $\wt{D}_{\max}^{1-\ve} (\rho \| \sigma) + \log \frac{1}{\ve} \leq D^{\ve}_H(\rho \| \sigma)$, which implies that $n R + n E_{\rm sc} + o(n) \leq D^{\exp\left(-nE_{\rm sc} + o(n)\right)}_H(\rho^{\otimes n} \| \sigma^{\otimes n})$. But then, the converse for the error exponent of hypothesis testing~\cite{nagaoka_2006,audenaert_2007} gives
\begin{equation}\begin{aligned}
E_{\rm sc} &\leq \sup_{\alpha\in(0,1)} \frac{\alpha-1}{\alpha} \Big( R + E_{\rm sc} - D_\alpha(\rho\|\sigma) \Big) + o(1) \\
&= \sup_{\alpha \in (0,1)} \frac1\alpha \Big( (\alpha-1) (R - D_\alpha(\rho\|\sigma) + (\alpha-1) E_{\rm sc} \Big) + o(1)\\
&\leq \sup_{\alpha \in (0,1)} \frac1\alpha \big( E_{\rm sc} + (\alpha-1) E_{\rm sc} - \delta \big) + o(1)\\
&= E_{\rm sc} - \delta + o(1),
\end{aligned}\end{equation}
which, in the limit $n\to\infty$, 
leads to a contradiction.}

For the smooth max-relative entropy, the bound from Corollary~\ref{cor:renyi_lower} does not give a tight lower bound on the strong converse exponent of~$\Dmax{P}{sub}$ in general~\cite{li_2024-1}. 
However, since we have seen the bound to be tight for the strong converse exponent of $\wt{D}^\ve_{\max}$, by Lemma~\ref{lem:classical_dmax_equivalence} it is also tight for the strong converse exponent of $\Dmax{T}{sub}$ when $\rho$ and $\sigma$ are commuting states --- this was already shown earlier in~\cite{salzmann_2022}.

{
We also remark that an application of Corollary~\ref{cor:renyi_lower} together with the tighter upper bound of Lemma~\ref{lem:purified_Dmax_lowerbound} gives another variant of a lower bound for the purified smoothing: for any $c \in (0,1)$, we have
\begin{align}\label{eq:Dmax_lower_alpha_purified}
    \Dmax{P}{sub}[\sqrt{c\ve}](\rho\|\sigma) \geq D_\alpha(\rho \| \sigma) - \frac{\alpha}{1-\alpha} \log \frac{1}{1- \ve} - \log \frac{1}{1-c}.
\end{align}

For clarity, let us give a formal statement of the asymptotic results discussed in this section.
}

\begin{boxed}
\begin{corollary}\label{cor:exponents}
{%
For any two states $\rho$, $\sigma$ and any $R > 0$, let
\begin{equation}\begin{aligned}
    \ve_n (R) \coloneqq \inf \lsetr \ve \in (0,1) \barr \wt{D}^\ve_{\max} (\rho^{\otimes n} \| \sigma^{\otimes n}) \leq n R \rsetr.
\end{aligned}\end{equation}
Then
\begin{equation}\begin{aligned}
    \lim_{n\to\infty} - \frac1n \log \ve_n(R) &= \sup_{\alpha > 1}\, (\alpha-1) \left( R - \wt{D}_\alpha(\rho \| \sigma) \right),\\
    \lim_{n\to\infty} - \frac1n \log \big(1-\ve_n(R)\big) &=\! \sup_{\alpha \in (0,1)} (\alpha-1)\big( R - D_\alpha (\rho \| \sigma) \big).
\end{aligned}\end{equation}
}%
\end{corollary}
\end{boxed}

{As discussed above, understanding the asymptotic properties of $\wt{D}^\ve_{\max}$ can find direct use in characterising $D^\ve_{\max}$, which then can be used to study tasks whose performance is connected to this quantity, such as quantum privacy amplification~\cite{li_2023,li_2024-1,salzmann_2022}. As another example of how the results here can be related with the analysis of the performance of operational problems, let us consider again the one-shot achievability bounds of~\cite{cheng_2023-1}. Recall from~\eqref{eq:haochung} that~\cite{cheng_2023-1} gives lower (achievable) bounds on code size of the form $\log M \geq \wt{D}^{1-\ve}_{\max}(\rho_{XB} \| \rho_X \otimes \rho_B)$ where $\ve$ represents the error. The error exponent of $\wt{D}^{\ve}_{\max}$ then corresponds to an achievability result for the strong converse exponent of the tasks studied in~\cite{cheng_2023-1}, and, vice versa, the strong converse exponent of $\wt{D}^\ve_{\max}$ gives achievability of the error exponents for the tasks. The bound of Corollary~\ref{cor:renyi_lower} then exactly matches the one-shot bound on the error exponent of classical-quantum channel coding given in~\cite[Proposition~1]{cheng_2023-1}, namely
\begin{equation}\begin{aligned}
    \log \frac1\ve \geq n \sup_{\alpha \in (0,1)} \,(\alpha-1) \left( \log M - D_\alpha (\rho_{XB} \| \rho_X \otimes \rho_B) \right).
\end{aligned}\end{equation}
On the other hand, the fact that the upper bound of Corollary~\ref{cor:var-mrel} is asymptotically achievable, as per~\cite[Theorem~IV.4]{mosonyi_2015}, then yields an achievable strong converse exponent:
\begin{equation}\begin{aligned}
    \log \frac{1}{1-\ve} \geq n \sup_{\alpha > 1}\, (\alpha-1) \left( \log M - \wt{D}_\alpha (\rho_{XB} \| \rho_X \otimes \rho_B) \right) + o(n).
\end{aligned}\end{equation}
This result did not appear in~\cite{cheng_2023-1} and in fact it was asked there whether the one-shot achievability result can be used to deduce achievable strong converse exponents, which we see to indeed be possible.  
Although the exponents obtained in this way do not match the optimal ones, the achievability result is significantly simpler than the more involved proof approaches needed to establish the optimal error~\cite{renes_2025,li_2025,cheng_2025} and strong converse exponents~\cite{mosonyi_2017} of classical-quantum channel coding. The fact that this achievability result is not sufficient to give the optimal exponential behaviour of the tasks can be seen from Lemma~\ref{lem:DH_Dtilde}:  if either $\varepsilon \sim \exp(-nE)$ or $1 - \ve \sim \exp(-nE)$ for some positive exponent $E$, there is always a non-vanishing gap between the asymptotic expansions of $\wt{D}^{1-\ve}_{\max}$ and $D^\ve_H$ due to the additive factor $\log\frac{1}{\ve(1-\ve)}$ in the lemma; however, the optimal asymptotic exponents of coding tasks  generally match those of $D^\ve_H$ . Whether one-shot approaches based on smooth relative entropies can be tightened to give optimal achievability is an interesting open problem.
}


{%
\subsection[Relation with information spectrum divergence\texorpdfstring{ $D^\ve_s$}{}]{Relation with information spectrum divergence $\boldsymbol{D^\ve_s}$}\label{sec:info_spec}

The divergence $\wt{D}^\ve_{\max}$ was originally introduced in~\cite{datta_2015} as an alternative to the \emph{information spectrum divergence}\footnote{Despite its commonly used name, $D^\ve_s$ is not actually a divergence, as it does not obey the data processing inequality.} $D^\ve_s$, with the latter defined as~\cite{tomamichel_2013}
\begin{equation}\begin{aligned}
    D^\ve_s(\rho\|\sigma) &= \log \sup \lset \gamma \bar \Tr \rho \{ \rho \leq \gamma \sigma \} \leq \ve \rset\\
                          &= \log \sup \lset \gamma \bar \Tr \rho \{ \rho > \gamma \sigma \} \geq 1- \ve \rset,
\end{aligned}\end{equation}
where $\{\rho \leq \gamma \sigma\}$ (respectively, $\{\rho > \gamma \sigma\}$) denotes the projector onto the eigenspace of $\gamma \sigma - \rho$ corresponding to the non-negative (respectively, negative) eigenvalues. The quantity $D^\ve_s$ was formalised in~\cite{tomamichel_2013} to study the precise relations between the information spectrum method~\cite{han_2003,nagaoka_2007} and quantum information-theoretic tasks such as quantum hypothesis testing. Here we improve on the known quantitative connections between the different functions.
}

{
\begin{boxed}
\begin{proposition}\label{lem:info_spec}
For all $\rho$, $\sigma$, all $\ve\in(0,1)$, and all $\mu \in (0,1-\ve)$, it holds that
\begin{equation}\begin{aligned}\label{eq:Ds_Dmax_bound}
  \wt{D}^{1-\ve}_{\max} (\rho\|\sigma) \leq D^{\ve}_s(\rho\|\sigma) \leq \wt{D}^{1-\ve-\mu}_{\max}(\rho\|\sigma) + \log\frac{1-\ve}{\mu}.
\end{aligned}\end{equation}
Consequently,
\begin{equation}\begin{aligned}\label{eq:Ds_DH_bound}
  D^{\ve}_s(\rho\|\sigma) + \log\frac{1}{1-\ve} \leq D^{\ve}_H(\rho\|\sigma) \leq D^{\ve+\delta}_{s}(\rho\|\sigma) + \log\frac1\delta.
\end{aligned}\end{equation}
\end{proposition}
\end{boxed}
}

{
The first of the inequalities in~\eqref{eq:Ds_Dmax_bound} is known~\cite[Proposition~4.3]{datta_2015}, as is the second inequality in~\eqref{eq:Ds_DH_bound}, originally observed in~\cite[Lemma~12]{tomamichel_2013}. 
Our upper bound of~\eqref{eq:Ds_Dmax_bound} and lower bound of~\eqref{eq:Ds_DH_bound} improve on the previous bounds by the additive factor of $\log\frac{1}{1-\ve}$.
\begin{proof}[Proof of Proposition~\ref{lem:info_spec}]
For the first inequality in~\eqref{eq:Ds_Dmax_bound}, take any $\gamma$ such that $\log \gamma > D^\ve_s$, so that $\Tr \rho \{\rho \leq \gamma \sigma\} > \ve$. This means that
\begin{equation}\begin{aligned}
    1-\ve &\geq \Tr \rho \{ \rho > \gamma \sigma \}
    \geq \Tr (\rho - \gamma\sigma) \{ \rho > \gamma \sigma \}
    = \Tr (\rho - \gamma \sigma)_+
\end{aligned}\end{equation}
and so $\gamma$ is a feasible solution for the minimisation in the definition of $\wt{D}^\ve_{\max}(\rho\|\sigma)$, yielding the claimed inequality upon taking an infimum over such $\gamma$.

This immediately implies the second inequality in~\eqref{eq:Ds_DH_bound} through the upper bound on $D^\ve_H$ stated in Lemma~\ref{lem:DH_Dtilde}.

For the second inequality in~\eqref{eq:Ds_Dmax_bound}, we adapt a proof approach of Hayashi from~\cite[Lemma~19]{hayashi_2016-2}. Take any $\lambda$ such that $\Tr (\rho - \lambda \sigma)_+ \leq 1-\ve-\mu$. Now, fix an arbitrary $\zeta \in (0,\mu)$, let $k = \frac{1-\ve}{\mu-\zeta}$ and assume towards a contradiction that $\gamma = k \lambda$ is a feasible solution for the maximisation in the definition of $D^\ve_s(\rho\|\sigma)$, meaning that $\Tr \rho \{ \rho > k \lambda \sigma \} \geq 1- \ve >0$. But
\begin{align}
  \frac{1-\ve-\mu}{1-\ve} \Tr \rho \{ \rho > k \lambda \sigma \} &< \frac{1-\ve-\mu+\zeta}{1-\ve} \Tr \rho \{ \rho > k \lambda \sigma \} \nonumber \\
  &= \left(1-\frac1k\right) \Tr \rho \{ \rho > k \lambda \sigma \} \nonumber \\
  &= \Tr \rho \{ \rho > k \lambda \sigma \} - \frac1k \Tr \rho \{ \rho > k \lambda \sigma \} \\
  &\leq \Tr \rho \{ \rho > k \lambda \sigma  \} - \Tr \lambda \sigma \{ \rho > k \lambda \sigma \} \nonumber \\
  &= \Tr (\rho - \lambda \sigma) \{ \rho > k \lambda \sigma  \} \nonumber \\
  &\leq \Tr (\rho - \lambda \sigma)_+ \nonumber \\
  &\leq 1-\ve-\mu, \nonumber
\end{align}
where the penultimate inequality follows from the variational form of $\Tr(\cdot)_+$~\eqref{eq:trace_variational}. 
We thus get that $\Tr \rho \{ \rho > k \lambda \sigma \} < 1-\ve $, contradicting our assumption of the feasibility of $\gamma = k\lambda$; hence it must be the case that $\log( k \lambda) \geq D^\ve_s(\rho\|\sigma)$, which upon minimisation over $\lambda$ and taking $\zeta \to 0$ gives  exactly~\eqref{eq:Ds_Dmax_bound}.

Taking an infimum of both sides of the second inequality in~\eqref{eq:Ds_Dmax_bound} over all $\mu \in (0, 1-\ve)$, by Theorem~\ref{thm:equivalence_DH_Dtilde} we get exactly the first inequality in~\eqref{eq:Ds_DH_bound}, which was the last missing piece in the proof.
\end{proof}
}%

{%
The result of Proposition~\ref{lem:info_spec} can help clarify the relations between various one-shot achievability bounds in quantum information, which often appear in terms of either $D^\ve_s$ or $D^\ve_H$ (or indeed in terms of $\wt{D}^{1-\ve}_{\max}$, as in~\eqref{eq:haochung}). For instance, we get that $\wt{D}^{1-\ve}_{\max} \geq D^{\ve-\delta}_H - \log\frac1\delta \geq D^{\ve-\delta}_{s} - \log\frac{1-\ve}{\delta} + O(\delta)$, tightening the known relations between the one-shot achievability result of~\cite{cheng_2023-1} and the prior bound of~\cite{beigi_2014}, of relevance to their higher-order asymptotic expansion~\cite[\S 3.1]{cheng_2023-1}.

Our bounds on $D^s_\ve$ in Proposition~\ref{lem:info_spec} are also tight at the level of the asymptotic exponents, when $\ve$ converges exponentially fast to 0 or to 1. Indeed, the asymptotics of $D^\ve_s$ are nothing but the asymptotics of the quantum Neyman--Pearson tests, well studied in hypothesis testing~\cite{nagaoka_2007,audenaert_2012,mosonyi_2015}. To see how these results are implied by our bounds, observe that combining the bounds of Proposition~\ref{lem:info_spec} with the data-processing upper bound on $D^\ve_H$ recalled in~\eqref{eq:DH_upper_Dalpha} gives
\begin{equation}\begin{aligned}
    \wt{D}^{1-\ve}_{\max} (\rho\|\sigma) \leq D^{\ve}_s(\rho\|\sigma) \leq \wt{D}_\alpha (\rho\|\sigma) + \frac{1}{\alpha-1} \log \frac{1}{1-\ve}.
\end{aligned}\end{equation}
As we discussed in Section~\ref{sec:bounds_renyi}, this upper bound is asymptotically tight for $\wt{D}^{1-\ve}_{\max}$ when $\ve \to 0$ by~\cite[Theorem~IV.4]{mosonyi_2015}, making it asymptotically tight also  for $D^{\ve}_s$. On the other hand, invoking~\eqref{eq:Ds_Dmax_bound} with the choice $\mu = (1-c) \ve$ for some $c \in (0,1)$ gives
\begin{equation}\begin{aligned}
    \wt{D}^{\ve}_{\max} (\rho\|\sigma) \leq D^{1-\ve}_s(\rho\|\sigma) \leq \wt{D}^{c \ve}_{\max}(\rho\|\sigma) + \log\frac{1}{1-c},
\end{aligned}\end{equation}
which means that the asymptotic behaviour of $D^{1-\ve}_s$ when $\ve \to 0$ matches that of $\wt{D}^\ve_{\max}$. The exponents given by Corollary~\ref{cor:exponents} thus hold also when $\wt{D}^\ve_{\max}$ is replaced by $D^{1-\ve}_s$.

Although in this work we have emphasised the close connections between $\wt{D}^\ve_{\max}$ and the smooth max-relative entropy $D^\ve_{\max}$, the insights discussed here strengthen also the interpretation of $\wt{D}^\ve_{\max}$ as a quantity equivalent, in a sense, to the information spectrum divergence: not only do $\wt{D}^{1-\ve}_{\max}$ and $D^\ve_s$ have the same second-order asymptotic expansion~\cite{datta_2015}, they also share the same asymptotic behaviour in terms of error and strong converse exponents and similar, tight one-shot bounds.
}


\subsection{Simultaneous smoothing} \label{sec:simultaneous}

Let us now consider a multi-partite generalisation of the Datta--Renner lemma (Theorem~\ref{tightened_DR_lemma}). This involves smoothing all marginals of a global state at the same time, and it is reminiscent of the `simultaneous smoothing' conjecture of Drescher and Fawzi~\cite[Conjecture~III.1]{drescher_2013}. While the authors there were only concerned with min-entropy smoothing, intrinsically linked to smoothing of the max-relative entropy with respect to the maximally mixed state, we consider the more general case where the second state is arbitrary. However, an important limitation of our result is that we achieve simultaneous smoothing over non-overlapping sub-systems only, while Drescher and Fawzi are interested in smoothing simultaneously over \emph{all} sets of sub-systems (cf.~\cite[Theorem~VI.2]{drescher_2013}).
The same joint smoothing setting that we consider here appeared in~\cite{anshu_2019}.

\begin{boxed}
\begin{prop} \label{multi_partite_DR_thm}
For some positive integer $m\in \NN_+$ and some choice of parameters $\ve_1,\ldots,\ve_m\in [0,1]$ with $\sum_j \ve_j < 1$, let $\rho_{1\ldots m}$ be an $m$-partite state such that $\rho_i \leq A_i + Q_i$ holds for all $i=1,\ldots, m$, for some choice of operators $A_i,Q_i\geq 0$ with $\Tr Q_i \leq \ve_i$.  Here, $\rho_i$ denotes the marginal of $\rho$ on the $i^\text{th}$ system. Then there exists an $m$-partite normalised state $\rho'_{1\ldots m}$ such that 
\bb
&\rho'_i \leq \frac{A_i}{1-\sum_j \ve_j}\quad \forall\ i=1,\ldots, m\\
\text{and}\quad& F(\rho,\rho') \geq 1 - \sumno_j \ve_j, \quad \frac12\norm{\rho-\rho'}{1} &\leq \sqrt{\sumno_j \ve_j}.
\ee
\end{prop}
\end{boxed}

Analogous bounds hold for smoothing involving a subnormalised state $\rho'_{1\ldots m}$ and  its marginals  (as in Theorem~\ref{tightened_DR_lemma})  but we omit them for brevity. 

\begin{proof}
We generalise the proof of Theorem~\ref{tightened_DR_lemma}. Setting $G_i \coloneqq A_i \# \big((A_i+Q_i)^{-1}\big)$ and $G\coloneqq \bigotimes_i G_i$, by monotonicity of the matrix geometric mean, as in~\eqref{gm_proof_eq2}, we have that $0\leq G_i\leq (A_i+Q_i) \# \big((A_i+Q_i)^{-1}\big) = \id$. By the same reasoning as in~\eqref{gm_proof_eq4}, $\Tr \rho \big(G^2_i \otimes \id_{i^c}\big) = \Tr \rho_i G^2_i \geq 1-\ve_i$, where $\id_{i^c}$ denotes the identity operator over all sub-systems except the $i^\text{th}$. Since $\big[G^2_i \otimes \id_{i^c},\,G^2_j \otimes \id_{j^c}\big] = 0$ for all $i,j$, it is not difficult to verify by induction that
\bb
G^2 = \bigotimes_i G^2_i \geq \id - \sum_i (\id - G^2_i) \otimes \id_{i^c}\, ,
\ee
from which it follows that
\bb
\Tr \rho G^2 \geq 1 - \sum_i (1 - \Tr \rho_i G^2_i) \geq 1 - \sum_i \ve_i\, .
\ee
See also a very similar argument in~\cite[Eq.~(1.3)]{khabbazioskouei_2019}. 

Since $0\leq G_i \leq \id$, by taking tensor products and squaring we infer that $0\leq G^2 \leq \id$. Thus, from the gentle measurement lemma we conclude that the post-measurement state
\bb
\rho' \coloneqq \frac{G\rho G}{\Tr G^2 \rho}
\ee
satisfies that $F(\rho,\rho') \geq 1 - \sumno_i \ve_i$.
Also,
\bb
\rho'_i = \frac{1}{\Tr  G^2\rho}\, G_i \Tr_{i^c}\!\left[ \rho \left(\bigotimes\nolimits_{j:\, j\neq i} G_j^2\right) \right] G_i \leq \frac{1}{\Tr G^2\rho}\, G_i \rho_i G_i \leq \frac{A_i}{\Tr G^2\rho} \leq \frac{A_i}{1-\sum_i \ve_i}\, ,
\ee
concluding the proof.
\end{proof}

Combined with our improvements in the weak/strong converse duality, analogously to Theorem~\ref{cor:wsc}, we obtain the following improvement over the bound between $D^\ve_{\max}$ and $D^{1-\ve}_H$ given in~\cite[Theorem~5]{anshu_2019}.

\begin{corollary}\label{cor:joint_smoothing}
Let ${\rho\vphantom{'}}_{1\ldots m}$ be an $m$-partite state with marginals $\rho_i$. For all states $\sigma_i$ acting on the $i^{\rm th}$ system, and for all $\ve_i \in [0,1)$ such that $\sum_i \ve_i < 1$, there exists a state $\rho'_{1\ldots m}$ such that
\begin{equation}\begin{aligned}
    \frac12\norm{{\rho\vphantom{'}}_{1\ldots m} - \rho'_{1\ldots m}}{1} \leq P({\rho\vphantom{'}}_{1\ldots m}, \rho'_{1\ldots m}) \leq \sqrt{\sum_i \ve_i}
\end{aligned}\end{equation}
and
\begin{equation}\begin{aligned}
    D_{\max}(\rho'_i \| \sigma_i) + \log \frac{1 - \sum_j \ve_j}{\ve_i (1-\ve_i)} \leq D^{1-\ve_i}_H(\rho_i \| \sigma_i) \quad \forall i \in \{1,\ldots,m\}.
\end{aligned}\end{equation}
\end{corollary}

\begin{proof}
Applying Proposition~\ref{multi_partite_DR_thm} with the choices of $A_i = \lambda_i \sigma_i$, where $\rho_i \leq \lambda_i \sigma_i + Q_i$ is any feasible solution for $\wt{D}^{\ve_i}_{\max}(\rho_i \| \sigma_i)$, guarantees the existence of the suitable state $\rho'$ satisfying
\begin{equation}\begin{aligned}
    D_{\max}(\rho'_i \| \sigma_i) \leq \wt{D}^{\ve_i}_{\max} (\rho_i \| \sigma_i) + \log\frac{1}{1- \sum_i \ve_i}.
\end{aligned}\end{equation}
From Lemma~\ref{lem:DH_Dtilde} we then get that $\wt{D}^{\ve_i}_{\max}(\rho_i \| \sigma_i) + \log\frac{1}{\ve_i (1-\ve_i)} \leq D^{1-\ve_i}_H(\rho_i \| \sigma_i)$ holds for each $i$.
\end{proof}

\subsection{Quantum substate theorem}

Another  result that can be immediately strengthened by using our bounds is the so-called {\em{quantum substate theorem}}~\cite{jain_2002,jain_2012}, which provides an  upper bound on the smooth max-relative entropy in terms of the quantum relative entropy $D(\rho\|\sigma)$.

\begin{boxed}
\begin{corollary}[(Tighter quantum substate theorem)]\label{cor:qsstate}
For all states $\rho, \sigma$, and all $\ve \in (0,1)$ we have
\begin{equation}\begin{aligned}
   \Dmax{T}{n}(\rho\|\sigma) \leq \Dmax{P}{n}(\rho\|\sigma) \leq \frac{D(\rho\|\sigma)+1}{\ve^2} - \log\frac{1}{\ve^2}.
\end{aligned}\end{equation}
\end{corollary}
\end{boxed}
\begin{proof}
From~\eqref{eq:wsc_purified} in Theorem~\ref{cor:wsc} we get that $\Dmax{P}{n}(\rho\|\sigma) + \log\frac{1}{\ve^2} \leq D^{1-\ve^2}_H(\rho\|\sigma)$, and
\begin{equation}\begin{aligned}
    D^{1-\ve^2}_H(\rho\|\sigma) \leq \frac{D_{\MM}(\rho\|\sigma)+h(\ve^2)}{\ve^2} \leq \frac{D(\rho\|\sigma)+1}{\ve^2}
\end{aligned}\end{equation}
is the standard weak converse bound in the quantum Stein's lemma~\cite{hiai_1991}, with $h(\ve^2)$ denoting the binary entropy function.
\end{proof}
This improves over the inequality between $\Dmax{P}{n}$ and $D(\rho\|\sigma)$ that follows from~\cite[Theorem~1]{jain_2012} by an additive term of $\log\frac{1}{\ve^2(1-\ve^2)}$. Note that the quantum substate theorem was originally stated in terms of the `observational divergence' $D_{\rm obs} (\rho \| \sigma) \coloneqq \sup \lset \Tr M \rho \log \frac{\Tr M \rho}{\Tr M \sigma} \bar 0 \leq M \leq \id \rset$. Using the dual form of $\wt{D}^\ve_{\max}$ in~\eqref{eq:Dtilde_dual_tighter} one can also straightforwardly show the relation $\wt{D}^\ve_{\max} (\rho \| \sigma) \leq \frac1\ve D_{\rm obs}(\rho\|\sigma)$, which recovers exactly the statement in~\cite[Theorem~1]{jain_2012} by Corollary~\ref{cor:DR_dmax}. However, this leads to a weaker bound between $\Dmax{P}{n}$ and $D(\rho\|\sigma)$.
Another variant of an upper bound involving $D(\rho\|\sigma)$ was also given in~\cite[Proposition~5]{wang_2019}.

\section{An alternative integral representation of the Umegaki relative entropy by rotating Frenkel's integral}
\label{sec:Frenkel}

Recently, a new integral representation of the Umegaki quantum relative entropy $D(\rho\|\sigma)$ was proposed by Frenkel~\cite{Frenkel2023}, sparking a wave of interest that resulted in new fundamental insights on quantum relative entropies and their properties~\cite{Jencova2023, Hirche2023, continuity-via-integral, Beigi2025}. Here, we show that it is possible to `rotate' Frenkel's integral formula so as to obtain a fundamentally new integral representation of the Umegaki relative entropy $D(\rho\|\sigma)$ that involves the modified version $\wt{D}^\ve_{\max} (\rho \| \sigma)$ of the smooth max-relative entropy (see~\eqref{Datta_Leditzky} for a definition). Our result is as follows.

\begin{boxed}
\begin{thm}
\label{thm:Frenkel-extn}
For all pairs of quantum states $\rho$ and $\sigma$ on the same finite-dimensional system, it holds that
\bb
D(\rho\|\sigma) &= \int_0^1 \dd\e\ \left(\wt{D}^\ve_{\max} (\rho \| \sigma) + (\log e) \left(1 - \exp\left[ - \wt{D}^\ve_{\max} (\rho \| \sigma) \right]\right) \right)_+ \\
&= \int_0^{\frac12 \|\rho-\sigma\|_1} \dd\e\ \left(\wt{D}^\ve_{\max} (\rho \| \sigma) + (\log e) \left(1 - \exp\left[ - \wt{D}^\ve_{\max} (\rho \| \sigma) \right]\right) \right) .
\ee
These identities hold for every choice of logarithm base, provided that $\exp$ is taken to be the inverse of the $\log$ function.
\end{thm}
\end{boxed}

\begin{proof}
It is clear that the claim holds for one choice of (positive) logarithm base if and only if it holds for all such choices. Therefore, we can choose without loss of generality the most convenient base to work with, i.e.\ the Euler's number $e$. 

If $\supp \rho \not\subseteq \supp \sigma$ then all expressions diverge --- in the integrals, we have that $\widetilde{D}_{\max}^\e(\rho\|\sigma) = +\infty$ for sufficiently small $\e$. Therefore, from now on we are going to assume that $\supp \rho \subseteq \supp \sigma$. We start from Frenkel's integral formula~\cite{Frenkel2023}, written here \`a la Hirche and Tomamichel~\cite[Corollary~2.3]{Hirche2023}:
\bb
D(\rho\|\sigma) &= \int_1^\infty \frac{\dd \gamma}{\gamma}\, \Tr (\rho - \gamma \sigma)_+ + \int_1^\infty \frac{\dd \gamma}{\gamma^2}\, \Tr (\sigma - \gamma\rho)_+ \\
&= \int_0^\infty \dd t\ \Tr \big(\rho - e^t \sigma\big)_+ + \int_0^1 \dd u\, \Tr \big(\sigma - \tfrac1u \rho\big)_+\, .
\ee
Here, in the second line we introduced two changes of variable: $\gamma = e^t$ for the first integral, and $\gamma = 1/u$ for the second. The function $[0,\infty) \ni t \mapsto \Tr \big(\rho - e^t \sigma\big)_+$ is monotonically non-increasing, it evaluates to $\Tr (\rho-\sigma)_+ = \frac12 \|\rho-\sigma\|_1$ for $t=0$, and it becomes identically $0$ for sufficiently large $t$. The integral $\int_0^\infty \dd t\ \Tr \big(\rho - e^t \sigma\big)_+$ is simply the area under its curve. This area can be computed also in another way, by adding up the areas of horizontal strips instead of vertical ones. For $0\leq \e\leq \frac12 \|\rho-\sigma\|_1$ define
\bb
f(\e) \coloneqq \min\left\{ t:\ \Tr \big(\rho - e^t \sigma\big)_+ \leq \e\right\} = \widetilde{D}^\e_{\max}(\rho\|\sigma) \geq 0\, .
\ee
This is just the inverse function to $t\mapsto \Tr \big(\rho - e^t \sigma\big)_+$. Then clearly
\bb
\int_0^\infty \dd t\ \Tr \big(\rho - e^t \sigma\big)_+ = \int_0^{\frac12 \|\rho-\sigma\|_1} \dd\e\ f(\e) = \int_0^{\frac12 \|\rho-\sigma\|_1} \dd\e\ \widetilde{D}^\e_{\max}(\rho\|\sigma)\, .
\ee

A similar reasoning can be repeated for the second integral. This time $[0,1] \ni u \mapsto \Tr \big(\sigma - \tfrac1u \rho\big)_+$ is monotonically non-decreasing, it evaluates to $\Tr (\sigma - \rho)_+ = \frac12 \|\rho-\sigma\|_1$ for $u=1$, and its inverse function is
\bb
g(\e) \coloneqq&\ \max\big\{ u\in [0,1]:\  \Tr \big(\sigma - \tfrac1u \rho\big)_+ \leq \e\big\} \\
=&\ \left(\min\big\{ s\in [1,\infty]:\  \Tr \big(\sigma - s \rho\big)_+ \leq \e\big\}\right)^{-1} \\
=&\ \exp\left[ - \widetilde{D}^\e_{\max}(\sigma\|\rho) \right] .
\ee
Hence,
\bb
\int_0^1 \dd u\, \Tr \big(\sigma - \tfrac1u \rho\big)_+ &= \int_0^{\frac12 \|\rho-\sigma\|_1} \dd \e\ \left(1-\exp\left[ - \widetilde{D}^\e_{\max}(\sigma\|\rho) \right]\right) .
\ee
The claim is proved once one adds up these two integrals, observing that both integrands are non-negative in the same range $0\leq \e\leq \frac12 \|\rho-\sigma\|_1$.
\end{proof}


\bigskip

\textbf{Acknowledgments.} --- We acknowledge helpful discussions with Mario Berta, Marco Tomamichel, and Mark M.\ Wilde. B.R. acknowledges the support of the Japan Science and Technology Agency (JST) PRESTO grant no.\ JPMJPR25FB. L.L. acknowledges support from MIUR (Ministero dell'Istruzione, dell'Universit\`a e della Ricerca) through the project `Dipartimenti di Eccellenza 2023--2027' of the `Classe di Scienze' department at the Scuola Normale Superiore and from the European Union under the European Research Council (ERC Grant
Agreement No.\ 101165230). N.D. is supported by the Engineering and Physical Sciences Research Council [Grant Ref: EP/Y028732/1].


 \bibliographystyle{apsca}
 \bibliography{main}

\appendix


\counterwithin{thm}{section}

\section[\texorpdfstring{Equivalent definitions of ${\wt{D}^\ve_{\max}}$}{Equivalent definitions of modified max-relative entropy}]{Equivalent definitions of $\boldsymbol{\wt{D}^\ve_{\max}}$}\label{app:definitions}

Recall that we followed~\cite{datta_2015,nuradha_2024} in defining
\begin{equation}\begin{aligned}
    \wt{D}^\ve_{\max} (\rho \| \sigma) &= \log \inf \lset \lambda \bar \rho \leq \lambda \sigma + Q,\; Q \geq 0,\; \Tr Q \leq \ve \rset\\
    &= \log \sup \lset \frac{\Tr W' \rho - \ve}{\Tr W' \sigma} \bar 0 \leq W' \leq \id \rset.
\end{aligned}\end{equation}
This also mirrors the definition of the classical $\ve$-approximate max-divergence in~\cite{dwork_2010}.

In~\cite{hirche_2023}, a quantity that we will call the \emph{max-relative entropy smoothed over unnormalised positive operators} was defined as
\begin{equation}\begin{aligned}
    D_{\max}^{\ve,\, \rm pos} (\rho \| \sigma) \coloneqq \log \inf \lset \lambda \bar Z \leq \lambda \sigma,\; Z \geq 0,\; \Tr( \rho - Z)_+ \leq \ve \rset.
\end{aligned}\end{equation}
The trace of the operators $Z$ here is completely unconstrained.

The \emph{max-relative entropy smoothed over Hermitian operators} made an inadvertent and implicit appearance in~\cite{zhou_2017} (cf.\ the discussion of the smoothing issue in~\cite{hirche_2023}). We can define it as
\begin{equation}\begin{aligned}
    D_{\max}^{\ve,\,\rm Herm,\, \leq} (\rho \| \sigma) \coloneqq \log \inf \lset \lambda \bar X \leq \lambda \sigma,\; X = X^\dagger,\; \Tr X \leq 1,\; \norm{\rho - X}{+} \leq \ve  \rset.
\end{aligned}\end{equation}
This resembles the definition of $\Dmax{T}{sub}$, but the operators $X$ are not required to be positive semidefinite.

\begin{boxed}
\begin{lemma}
For all quantum states and all $\ve \in [0,1]$, it holds that
\begin{equation}\begin{aligned}
    D_{\max}^{\ve,\, \rm pos} (\rho \| \sigma) =  D_{\max}^{\ve,\,\rm Herm,\, \leq} (\rho \| \sigma) =  \wt{D}^\ve_{\max} (\rho \| \sigma).
\end{aligned}\end{equation}
\end{lemma}
\end{boxed}
The equivalence of $D_{\max}^{\ve,\, \rm pos}$ with $\wt{D}^\ve_{\max}$ can already be deduced from~\cite[Lemma II.6]{hirche_2023}, which showed that $D_{\max}^{\ve,\, \rm pos}$ corresponds to the inverse function of the hockey-stick divergence $E_\lambda (\rho\|\sigma) = \Tr (\rho-\lambda \sigma)_+$. More generally, all of the quantities can be deduced to be equivalent because they all have been shown to tightly characterise quantum approximate differential privacy~\cite{zhou_2017,hirche_2023,nuradha_2024}. Notwithstanding, we give a concise direct proof.
\begin{proof}
Consider any feasible solution for $\wt{D}^\ve_{\max}$ in the form of $\rho \leq \lambda \sigma + Q$ with $Q \geq 0$ and $\Tr Q \leq \ve$. Then the operator $Z = \lambda\sigma$  satisfies $Z \leq \lambda \sigma$ and $\rho - Z \leq Q$, hence $\Tr (\rho - Z)_+ \leq \Tr Q \leq \ve$, implying that $D_{\max}^{\ve,\, \rm pos} \leq \log \lambda$. Furthermore, the operator $X = \rho - Q$ satisfies $X \leq \lambda \sigma$, $\Tr X \leq \Tr \rho = 1$ and $\norm{\rho - X}{+} = \norm{Q}{+} \leq \ve$, and hence $D_{\max}^{\ve,\,\rm Herm,\, \leq} \leq \log \lambda$.

On the other hand, consider first a feasible solution for $D_{\max}^{\ve,\, \rm pos}$, namely an operator $Z\geq 0$ such that $Z \leq \lambda \sigma$ and $\Tr (\rho - Z)_+ \leq \ve$. Then $\rho \leq Z + (\rho - Z)_+ \leq \lambda \sigma + (\rho - Z)_+$, so defining $Q = (\rho - Z)_+$ we have a feasible solution for $\wt{D}^\ve_{\max}$. Very similarly, any feasible solution $ D_{\max}^{\ve,\,\rm Herm,\, \leq}$, that is, a Hermitian operator $X$ satisfying $X \leq \lambda \sigma$ and $\norm{\rho - X}{+} \leq \ve$, gives a feasible solution for $\wt{D}^\ve_{\max}$ as $\rho \leq \lambda \sigma + (\rho - X)_+$, where $\Tr (\rho - X)_+ \leq \ve$ by definition of the generalised trace distance.
\end{proof}

Note that, in the classical case, all of these smoothing variants reduce to the standard smooth max-relative entropy $\Dmax{T}{sub}$, as per Lemma~\ref{lem:classical_dmax_equivalence}.

One can also define a normalised variant of the Hermitian-smoothed max-relative entropy, $D_{\max}^{\ve,\,\rm Herm,\, =}$, and follow the proof of Lemma~\ref{lem:subnormalised_normalised_main} to show that
\begin{equation}\begin{aligned}
    D_{\max}^{\ve,\,\rm Herm,\, =} (\rho \| \sigma) = \max \left\{  D_{\max}^{\ve,\,\rm Herm,\, \leq} (\rho \| \sigma),\; 0 \right\} = \max \left\{  \wt{D}^\ve_{\max} (\rho \| \sigma),\; 0 \right\}.
\end{aligned}\end{equation}


\section[\texorpdfstring{Continuity and strict monotonicity properties of ${\wt{D}^\ve_{\max}}$}{Continuity and strict monotonicity properties of modified max-relative entropy}]{Continuity and strict monotonicity properties of $\boldsymbol{\wt{D}^\ve_{\max}}$}\label{app:continuity} 

{%
This appendix is devoted to the investigation of some analytical properties of the function $\e\mapsto \widetilde{D}_{\max}^\e(\rho\|\sigma)$ for any two fixed states $\rho,\sigma$. In what follows, we denote with $\Pi_\sigma$ the projector 
onto the support of a state $\sigma$, and with $\{\rho \geq \gamma \sigma\}$ the projector onto the positive part of the operator $\rho - \gamma \sigma$. We start with the following simple result.

\begin{lemma} \label{lem:convergence_projector}
Let $\rho$ and $\sigma$ be two finite-dimensional quantum states. Then
\bb
\lim_{\gamma\to\infty} \{\rho \geq \gamma \sigma\} = \id - \Pi_\sigma.
\ee
\end{lemma}

\begin{proof}
If $\Pi_\sigma = \id$, so that $\sigma$ has no zero eigenvalues, then $\rho \leq \gamma \sigma$ for all sufficiently large $\gamma$, and there is nothing to prove. Therefore, let us assume that $\sigma$ has at least one zero eigenvalue. 

Setting $x\coloneqq 1/\gamma$, we have $\{\rho \geq \gamma \sigma\} = \{ \sigma - x \rho \leq 0\}$. Imagine sorting the spectra of $\sigma$ and $\sigma - x \rho$ in (say) non-increasing order, repeating the eigenvalues according to their multiplicities. By~\cite[Theorem~5.1 and Eq.~(5.2)]{KATO}, every eigenvalue of $\sigma - x \rho$ must converge to the corresponding eigenvalue of $\sigma$ as $x\to 0$ (note that~\cite[Eq.~(5.2)]{KATO} prohibits changes in multiplicities). Now, consider an interval $I = [-\delta,\delta]$, so that $\delta$ is strictly smaller than the smallest non-zero eigenvalue of $\sigma$. Invoking again~\cite[Theorem~5.1]{KATO}, we see that the total projector $P_I(x)$ onto the span of all the eigenspaces of $\sigma - x\rho$ whose eigenvalues lie in $I$ is a continuous function of $x$ at $x=0$. That is, $P_I(0) = \lim_{x\to 0^+} P_I(x) = \id - \Pi_\sigma$. By Weyl's monotonicity theorem~\cite[Corollary~III.2.3]{bhatia_1996}, and sorting the spectra in non-increasing order as before, every eigenvalue of $\sigma - x\rho$ must be smaller than or equal to the corresponding eigenvalue of $\sigma$, when $x>0$. This means that $P_I(x) = \{\sigma - x\rho \leq 0\}$ for all sufficiently small $x>0$. Therefore, $\lim_{x\to 0^+} \{\sigma - x\rho \leq 0\} = \lim_{x\to 0^+} P_I(x) = P_I(0) = \id - \Pi_\sigma$.
\end{proof}

We are now ready to prove the following.

\begin{boxed}
\begin{lemma} 
Let $\rho$ and $\sigma$ be two finite-dimensional quantum states. Then:
\begin{enumerate}[(a)]
\item If $\big[\rho,\Pi_\sigma\big] \neq 0$, the function $\gamma \mapsto \Tr (\rho - \gamma \sigma)_+$ is continuous and strictly decreasing on the whole half-positive real line $\gamma\in [0,\infty)$; furthermore, it holds that
\bb
\lim_{\gamma\to\infty} \Tr (\rho - \gamma\sigma)_+ = 1 - \Tr \rho \Pi_\sigma\, .
\ee

\item If $\big[\rho,\Pi_\sigma\big] = 0$, the function $\gamma \mapsto \Tr (\rho - \gamma \sigma)_+$ is continuous, strictly decreasing on the interval $\gamma\in \big[0,\,\exp[D_{\max}(\rho \Pi_\sigma \|\sigma)] \big]$, and equal to $1 - \Tr \rho \Pi_\sigma$ for all $\gamma \geq \exp[D_{\max}(\rho \Pi_\sigma \|\sigma)]$.
\end{enumerate}
Consequently:
\begin{enumerate}[(a)] \setcounter{enumi}{2}
\item If $\big[\rho,\Pi_\sigma\big] \neq 0$, the function $\e \mapsto \widetilde{D}_{\max}^\e(\rho\|\sigma) \in [0,+\infty]$ is equal to $+\infty$ for $\e\in \big[0,\, 1 - \Tr \rho \Pi_\sigma \big]$, continuous and strictly decreasing for $\e\in \big(1 - \Tr \rho \Pi_\sigma,\, 1 \big]$, and such that
\bb
\lim_{\e \to \left(1 - \Tr \rho \Pi_\sigma\right)^+} \widetilde{D}_{\max}^\e(\rho\|\sigma) = +\infty\, .
\ee

\item If $\big[\rho,\Pi_\sigma\big] = 0$, the function $\e \mapsto \widetilde{D}_{\max}^\e(\rho\|\sigma)\in [0,+\infty]$ is equal to $+\infty$ if $\e\in \big[0,\, 1 - \Tr \rho \Pi_\sigma \big)$ (this latter interval is understood to be empty if $\supp(\rho)\subseteq \supp(\sigma)$), continuous and strictly decreasing for $\e\in \big[1 - \Tr \rho \Pi_\sigma,\, 1 \big]$, and such that
\bb
\lim_{\e \to \left(1 - \Tr \rho \Pi_\sigma\right)^+} \widetilde{D}_{\max}^\e(\rho\|\sigma) = D_{\max}\big(\rho \Pi_\sigma\big\|\sigma\big)\, .
\ee

\item In particular, $\e \mapsto \widetilde{D}_{\max}^\e(\rho\|\sigma) \in [0,+\infty]$ is always continuous from the right.
 
\end{enumerate}
\end{lemma}
\end{boxed}

\begin{proof}
The continuity of the function $\gamma\mapsto \Tr (\rho - \gamma \sigma)_+ \eqqcolon E_\gamma(\rho\|\sigma)$ follows immediately from its convexity~\cite[Lemma~2.1(9)]{Hirche2023}. The fact that this function is monotonically non-increasing is also well known~\cite[eq.~between Eqs.~(17)--(18)]{nagaoka_2007},
but it is helpful to reproduce the proof of this property here. Assume that $\gamma \leq \gamma'$. Then, as a consequence of the variational form of $\Tr(\cdot)_+$~(Eq.~\eqref{eq:trace_variational}),
\bb
E_\gamma(\rho\|\sigma) \geq \Tr (\rho - \gamma \sigma) \{\rho \geq\gamma' \sigma\} = E_{\gamma'}(\rho\|\sigma) + (\gamma'-\gamma) \Tr \sigma \{\rho \geq \gamma'\sigma\} \geq E_{\gamma'}(\rho\|\sigma)\, .
\label{eq:E_gamma_decreasing}
\ee
If $\gamma < \gamma'$, equality is only possible when $\Tr \sigma \{\rho \geq \gamma'\sigma\} = 0$. But then, for any $\gamma''\geq \gamma'$,
\bb
E_{\gamma''}(\rho\|\sigma) \geq \Tr (\rho - \gamma''\sigma) \{\rho \geq \gamma'\sigma\} = \Tr \rho \{\rho \geq \gamma'\sigma\} = E_{\gamma'}(\rho\|\sigma)\, ,
\ee
i.e.\ the function cannot decrease any more. When this happens, the function has reached the value
\bb
\lim_{\gamma\to\infty} E_\gamma(\rho\|\sigma) = \Tr \rho (\id - \Pi_\sigma) = 1 - \Tr \rho \Pi_\sigma\, .
\label{eq:limit_E_gamma}
\ee
To see why the first equality holds, observe that, on the one hand,
\bb
E_\gamma(\rho\|\sigma) = \sup_{0\leq Q\leq \id} \Tr (\rho - \gamma \sigma) Q \geq \Tr (\rho - \gamma \sigma) (\id - \Pi_\sigma) = 1 - \Tr \rho \Pi_\sigma\, ,
\ee
while, on the other,
\bb
E_\gamma(\rho\|\sigma) = \Tr (\rho - \gamma \sigma) \{\rho \geq \gamma \sigma\} \leq \Tr \rho \{\rho \geq \gamma \sigma\}\, ,
\ee
and the rightmost side tends to $1 - \Tr \rho \Pi_\sigma$ as $\gamma\to \infty$, due to Lemma~\ref{lem:convergence_projector}. This proves~\eqref{eq:limit_E_gamma}.

We saw that, when $\gamma < \gamma'$, equality in~\eqref{eq:E_gamma_decreasing} can hold only if $\Tr \sigma \{\rho \geq \gamma'\sigma\} = 0$. Since $\sigma$ and $\{\rho \geq \gamma'\sigma\}$ are positive semi-definite, this can happen if and only if their supports are orthogonal, entailing in particular that $[\sigma, \{\rho \geq \gamma'\sigma\}] = 0$. From this we deduce that
\bb
0 = \big[ \rho - \gamma'\sigma,\, \{\rho \geq \gamma'\sigma\} \big] = \big[ \rho,\, \{\rho \geq \gamma'\sigma\} \big] \tends{}{\gamma'\to \infty} [ \rho,\, \id - \Pi_\sigma ] = - [\rho,\Pi_\sigma]\, .
\ee

Hence, if $[\rho,\Pi_\sigma] \neq 0$, the function $\gamma \mapsto E_\gamma(\rho\|\sigma)$ must be strictly decreasing on the whole half-positive real line, proving~(a). If, on the contrary, $[\rho,\Pi_\sigma] = 0$, then we can decompose $\rho = \rho \Pi_\sigma + \rho (\id - \Pi_\sigma)$, where both addends are (Hermitian and) positive semi-definite. Hence, for an arbitrary $\gamma\geq 0$,
\bb
E_\gamma(\rho\|\sigma) = \Tr (\rho - \gamma \sigma)_+ = \Tr \rho (\id - \Pi_\sigma) + \Tr (\rho \Pi_\sigma - \gamma \sigma)_+ = 1 - \Tr \rho \Pi_\sigma + \Tr (\rho \Pi_\sigma - \gamma \sigma)_+\, ;
\ee
this is equal to $1 - \Tr \rho \Pi_\sigma$ if and only if $\rho \Pi_\sigma \leq \gamma \sigma$, i.e.\ if and only if $\gamma \geq \exp\left[D_{\max}(\rho \Pi_\sigma\|\sigma)\right]$. This completes the proof of~(b).

Claims~(c)--(e) follow from easily from~(a)--(b), once one remembers Eq.~\eqref{Datta_Leditzky}, which lets us express the function $\e \mapsto \widetilde{D}_{\max}^\e(\rho\|\sigma) = \log \inf\lset \gamma\geq 0 \bar E_\gamma(\rho\|\sigma)\leq \e \rset$ as a generalised inverse of the function $\gamma \mapsto E_\gamma(\rho\|\sigma)$.
\end{proof}
}%

\section{A gentler measurement lemma}\label{app:gentle}

{
\renewcommand{\thethm}{\ref{tighter_gentle_measurement_lemma}}
\begin{boxed}
\begin{lemma}
Let $M\in [0,\id]$ be a measurement operator and $\rho$ be an arbitrary state on the same system. If $\Tr M\rho \geq 1-\e$ for some $\e\in [0,1]$, then
\begin{align}
F\Big(\rho,\, \sqrt{M}\rho\sqrt{M}\Big) &\geq (1-\e)^2\, , \label{gentle_measurement_fidelity_app} \\
F\left(\rho,\, \frac{\sqrt{M}\rho \sqrt{M}}{\Tr M\rho}\right) &\geq 1-\e\, , \label{gentle_measurement_fidelity_norm_app} \\
\frac12 \left\|\rho - \frac{\sqrt{M}\rho \sqrt{M}}{\Tr M\rho} \right\|_1 &\leq \sqrt\e\, , \label{gentle_measurement_trace_distance_app} \\
\frac12 \left\|\rho - \sqrt{M}\rho\sqrt{M}\right\|_1 &\leq \begin{cases} \sqrt{\e\left(1-\frac{3\e}{4}\right)} & \text{\rm if $\e\leq 2/3$,} \\[1ex] 1/\sqrt3 & \text{\rm if $\e>2/3$} \end{cases}{\color{green}}\leq \sqrt{\ve} \,, \label{gentle_measurement_unnormalised_trace_distance_app} \\
 \left\|\rho - \sqrt{M}\rho\sqrt{M} \hspace{1pt}\right\|_+ &\leq \sqrt{\e\left(1-\frac{3\e}{4}\right)} + \frac\e2 \leq \sqrt{\ve(2-\ve)} \, . \label{gentle_measurement_gen_trace_distance_app}
\end{align}
All of the primary bounds are tight for all $\e\in [0,1]$.
\end{lemma}
\end{boxed}
}

\begin{proof}
The two numbers $x\coloneqq \Tr \rho M$ and $y\coloneqq \Tr \rho \sqrt{M}$ satisfy the following constraints: first, $x\in [1-\e,1]$; second, $x\leq y$, because $M\leq \sqrt{M}$ as the eigenvalues of $M$ are between $0$ and $1$; third, $y\leq \sqrt{x}$, which follows easily from the Cauchy--Schwarz inequality. These constraints imply that
\bb
\sqrt{F}\Big(\rho,\, \sqrt{M}\rho\sqrt{M}\Big) &= \Tr \sqrt{\sqrt{\rho} \sqrt{M}\rho\sqrt{M} \sqrt{\rho}} = \Tr \sqrt{\rho} \sqrt{M} \sqrt{\rho} \\
&= \Tr \sqrt{M} \rho = y \geq x \geq 1-\e\, 
\ee
and analogously
\bb
\sqrt{F}\left(\rho,\, \frac{\sqrt{M}\rho \sqrt{M}}{\Tr M\rho}\right) &= \frac{\Tr \sqrt{\sqrt{\rho} \sqrt{M}\rho\sqrt{M} \sqrt{\rho}}}{\sqrt{\Tr M \rho}}  \geq \frac{x}{\sqrt{x}} \geq \sqrt{1-\e}\,. 
\ee
This proves~\eqref{gentle_measurement_fidelity_app} and~\eqref{gentle_measurement_fidelity_norm_app}, which in turn implies~\eqref{gentle_measurement_trace_distance_app} immediately due to the Fuchs--van de Graaf inequalities~\eqref{eq:fvdg}. We therefore move on to the proof of~\eqref{gentle_measurement_unnormalised_trace_distance_app} and~\eqref{gentle_measurement_gen_trace_distance_app}.

Defining the two purifications of $\rho$ and $\sqrt{M}\rho\sqrt{M}$ respectively given by $\ket{\alpha} \coloneqq (\sqrt{\rho} \otimes \id )\ket{\phi}$ (normalised) and $\ket{\beta}\coloneqq \left(\sqrt{\sqrt{M}\rho\sqrt{M}} \otimes \id\right)\ket{\phi}$ (subnormalised), where $\ket{\phi}\coloneqq \sum_{i=1}^d \ket{ii}$ is the un-normalised maximally entangled state, by data processing we see that
\bb
\frac12 \left\|\rho - \sqrt{M}\rho\sqrt{M}\right\|_1 &\leq \frac12 \left\|\alpha - \beta \right\|_1 \\
&= \sqrt{\left( \frac{\Tr[\alpha+\beta]}{2} \right)^2 - \Tr \alpha \beta} \\
&= \sqrt{\left(\frac{1+x}{2}\right)^2 - y^2}\, .
\ee
Here, in the first line we introduced the notation $\alpha \coloneqq \ketbra{\alpha}$ and similarly for $\beta$, and the second one follows from the general formula 
\begin{equation}\begin{aligned}
    \|X\|_1 = \sqrt{2 \Tr(X^2) - (\Tr X)^2}
\end{aligned}\end{equation}
valid for all rank-two Hermitian operators with non-positive determinant (cf.\ proof of Lemma~\ref{lem:purified_Dmax_lowerbound}).
All that remains to do it to prove that
\bb
\max_{x\in [1-\e,1],\ x\leq y\leq \sqrt{x}} \sqrt{\left(\frac{1+x}{2}\right)^2 - y^2} = f(\e) \coloneqq \begin{cases} \sqrt{\e\left(1-\frac{3\e}{4}\right)} & \text{if $\e\leq 2/3$,} \\ 1/\sqrt3 & \text{if $\e>2/3$} \end{cases} .
\ee
The maximum is always achieved when $y=x$. The rest of the calculation is left to the reader.

The case of the generalised trace distance is analogous. We have
\bb
\left\|\rho - \sqrt{M}\rho\sqrt{M}\right\|_+ &\leq \frac12 \left\|\alpha - \beta \right\|_1 + \frac12 \left| \Tr( \alpha - \beta) \right|\\
&=  \sqrt{\left(\frac{1+x}{2}\right)^2 - y^2} + \frac12 (1-x)
\ee
using now the data processing inequality for $\|\cdot\|_+$~\cite[Prop.~3.8]{tomamichel_2016}. An explicit maximisation over the feasible range of $x$ and $y$ yields the stated result.

To see that the bounds are tight, it suffices to consider the single-qubit state $\rho = \ketbra{0}$ and $M = \ketbra{\phi_\e}$, where $\ket{\phi_\e} \coloneqq \sqrt{1-\e} \ket{0} + \sqrt{\e} \ket{1}$. Then on the one hand we have that $\Tr \rho M = 1-\e$, while on the other
\bb
F\Big(\rho,\, \sqrt{M}\rho\sqrt{M}\Big) = \Tr \rho \sqrt{M} = |\braket{0|\phi_\e}|^2 = 1-\e\, ,
\ee
and moreover
\bb
\frac12 \left\|\rho - \sqrt{M}\rho\sqrt{M}\right\|_1 &= \frac12 \left\| \ketbra{0} - (1-\e)\,\ketbra{\phi_\e} \right\|_1 = \sqrt{\e\left(1-\frac{3\e}{4}\right)}\, ,\\
\left\|\rho - \sqrt{M}\rho\sqrt{M}\right\|_+  &= \sqrt{\e\left(1-\frac{3\e}{4}\right)} + \frac\e2.
\ee
When $\e > 2/3$, to achieve the tighter bound for $\frac12\norm{\cdot}{1}$ it suffices to choose $\rho = \ketbra{0}$ and $M = \ketbra{\phi_{2/3}}$. This completes the proof.
\end{proof}


\section{Hilbert projective metric}\label{app:hilbert}

As a curiosity for interested readers who made it all the way to Appendix~\ref*{app:hilbert},  we can show that the weak/strong converse duality studied in this work extends beyond the smooth max-relative entropy to the smooth Hilbert projective metric $D_\Omega$, defined as~\cite{bushell_1973,reeb_2011}
\begin{equation}\begin{aligned}
    D_{\Omega}^{\smoothing{\Delta}{\ve}} (\rho\|\sigma) &\coloneqq \inf_{\rho' \in \B^\ve_{\Delta}(\rho)} D_\Omega(\rho' \| \sigma),\\
    D_\Omega(\rho \| \sigma) &\coloneqq D_{\max}(\rho\|\sigma) + D_{\max}(\sigma\|\rho),
\end{aligned}\end{equation}
where  $\Delta$ stands for one of $(\sTnorm)$, $(\sTsub)$, $(\sPnorm)$, $(\sPsub)$. 

The connection relies on a technical lemma that connects $D_\Omega^{\smoothing{\Delta}{\ve}}$ with $D_{\max}^{\smoothing{\Delta}{\ve}}$.
One can easily notice that $D_\Omega$ is typically much larger than $D_{\max}$, and indeed $D_{\Omega}(\rho \| \sigma) = \infty$ unless $\rho$ and $\sigma$ both have the same support.
Perhaps surprisingly, smoothing makes this difference (asymptotically) negligible. The finding below was originally shown to hold in the asymptotic i.i.d.\ setting~\cite{regula_2023}, although it is not difficult to extract the following one-shot statement from the proof in~\cite{regula_2023}.

\begin{lemma}\label{lem:hilbert} For all $\ve \in (0,1)$ and all $\eta \in (0,\ve)$, it holds that
\begin{equation}\begin{aligned}
    D^{\s{T}{n}[\ve+\eta]}_{\Omega}(\rho \| \sigma) - \log\frac{1}{\eta} \leq \Dmax{T}{n}(\rho \| \sigma) \leq D^{\s{T}{n}}_{\Omega}(\rho \| \sigma),\\
     D^{\s{P}{n}[\ve+\eta]}_{\Omega}(\rho \| \sigma) - \log\frac{1}{\eta^2} \leq \Dmax{P}{n}(\rho \| \sigma) \leq D^{\s{P}{n}}_{\Omega}(\rho \| \sigma).
\end{aligned}\end{equation}
\end{lemma}
\begin{proof}
The upper bound is obvious from the definition. For the lower bound, consider trace distance first and let $\rhos$ be any feasible state for $\Dmax{T}{n}(\rho\|\sigma)$, that is, one such that $\frac12\norm{\rho - \rhos}{1} \leq \ve$ and $\rhos \leq \lambda \sigma$. Define now
\begin{equation}\begin{aligned}
    \rho'' \coloneqq (1-\eta) \rhos + \eta \, \sigma.
\end{aligned}\end{equation}
Noticing that
\begin{equation}\begin{aligned}
    \frac12\norm{\rho'' - \rho}{1} &\leq \frac12\norm{\rho'' - \rho'}{1} + \frac12\norm{\rho' - \rho}{1} \leq \eta + \ve,
\end{aligned}\end{equation}
we have that $\rho''$ is a feasible state for $D^{\s{T}{n}[\ve+\eta]}_\Omega(\rho\|\sigma)$. 
We then bound
\begin{equation}\begin{aligned}
    \rho'' \leq (1-\eta) \lambda \sigma + \eta \sigma \leq (1-\eta) \lambda \sigma + \eta \lambda \sigma = \lambda \sigma,
\end{aligned}\end{equation}
where we simply used that $\lambda \geq 1$, and
\begin{equation}\begin{aligned}
    \sigma \leq \sigma + \frac{1-\eta}{\eta} \rhos = \frac{1}{\eta} \rho''.
\end{aligned}\end{equation}
This altogether gives
\begin{equation}\begin{aligned}
    D_{\Omega}(\rho'' \| \sigma) = D_{\max}(\rho''\|\sigma) + D_{\max}(\sigma \| \rho'') \leq \log \lambda + \log \frac{1}{\eta}.
\end{aligned}\end{equation}
Optimising over all feasible $\rho'$ concludes the proof of the trace distance case.

The statement for the purified distance is obtained similarly. Given a feasible state $\rho'$ with $P(\rho,\rho')\leq \ve$ we now construct
\begin{equation}\begin{aligned}
    \rho'' \coloneqq (1-\eta^2) \rhos + \eta^2 \, \sigma,
\end{aligned}\end{equation}
which satisfies that
\begin{equation}\begin{aligned}
    P(\rho'', \rho) \leq P(\rho'',\rho') + \ve
\end{aligned}\end{equation}
by the triangle inequality. Now, we can use the operator monotonicity of the square root to get
\begin{equation}\begin{aligned}
    F(\rho'', \rho') &\geq (1-\eta^2) \, F\left( \rho', \rho' \right) = 1-\eta^2
\end{aligned}\end{equation}
which gives $P(\rho'', \rho') \leq \eta$. The rest of the proof is analogous.
\end{proof}
Combined with Theorem~\ref{cor:wsc}, we obtain  a type of weak/strong converse duality relation that relates $D^{1-\ve}_H$ directly with $D^{\ve}_{\Omega}$.

\begin{boxed}
\begin{corollary}\label{cor:hilbert_wsc}
    For all $\ve \in (0,1)$, all $\eta \in (0,1-\sqrt{\ve})$, and all $\mu \in (0,\ve)$ it holds that
\begin{equation}\begin{aligned}
    D^{\s{T}{n}[\sqrt\ve+\eta]}_{\Omega}(\rho \| \sigma) - \log\frac{1}{\eta} + \log\frac{1}{\ve}\leq D^{1-\ve}_H(\rho \| \sigma) \leq D^{\s{T}{n}[\ve-\mu]}_{\Omega}(\rho \| \sigma) + \log\frac{1}{\mu},\\
       D^{\s{P}{n}[\sqrt\ve+\eta]}_{\Omega}(\rho \| \sigma) - \log\frac{1}{\eta^2} + \log\frac{1}{\ve}\leq D^{1-\ve}_H(\rho \| \sigma) \leq D^{\s{P}{n}[\sqrt{\ve-\mu}]}_{\Omega}(\rho \| \sigma) + \log\frac{1}{\mu^2}.
\end{aligned}\end{equation}
\end{corollary}
\end{boxed}

The result of Lemma~\ref{lem:hilbert} furthermore implies that not only do the smooth max-relative entropy and the smooth Hilbert projective metric have the same first-order asymptotics~\cite{regula_2023}, also their second-order asymptotic expansion is the same --- for $D^{\s{P}{n}[\ve]}_{\Omega}$, it is given by the result of~\cite{tomamichel_2013}.


\end{document}

%% file: macros.tex
\usepackage{newpxtext,newpxmath}

\usepackage[utf8]{inputenc}
\usepackage{amsthm}
\usepackage{amssymb}
\usepackage{amsmath}
\usepackage{bm,bbm}
\usepackage{braket}
\usepackage{dsfont}
\usepackage{mathdots}
\usepackage{mathtools}
\usepackage{enumerate}
\usepackage[shortlabels]{enumitem}
\usepackage{csquotes}
\usepackage{stmaryrd}
\usepackage{graphicx}
\usepackage{stackengine}
\usepackage{scalerel}
\usepackage{tensor}       
\usepackage[dvipsnames,svgnames]{xcolor}
\usepackage{xfrac}
\usepackage{centernot}
\usepackage{comment}
\usepackage{chngcntr}
\usepackage{booktabs}
\usepackage{tabularray}
\UseTblrLibrary{booktabs}

\usepackage[pdftex]{hyperref}
\hypersetup{
    bookmarksnumbered=true, 
    breaklinks=true,
    unicode=false, 
    pdfstartview={FitH}, 
    pdfnewwindow=true, 
    colorlinks=true, 
    allcolors=MidnightBlue!70!black!70!TealBlue
}

\newtheorem{thm}{Theorem}
\newtheorem*{thm*}{Theorem}
\newtheorem{theorem}[thm]{Theorem}
\newtheorem*{theorem*}{Theorem}
\newtheorem{prop}[thm]{Proposition}
\newtheorem*{prop*}{Proposition}
\newtheorem{proposition}[thm]{Proposition}
\newtheorem*{proposition*}{Proposition}
\newtheorem{lemma}[thm]{Lemma}
\newtheorem*{lemma*}{Lemma}

\newtheorem{corollary}[thm]{Corollary}
\newtheorem*{cor*}{Corollary}

\newtheorem*{cj*}{Conjecture}

\newtheorem*{Def*}{Definition}

\newtheorem*{question*}{Question}

\newtheorem*{problem*}{Problem}

\makeatletter
\def\thmhead@plain#1#2#3{%
  \thmname{#1}\thmnumber{\@ifnotempty{#1}{ }\@upn{#2}}%
  \thmnote{ {\the\thm@notefont#3}}}
\let\thmhead\thmhead@plain
\makeatother

\theoremstyle{definition}

\newtheorem*{rem*}{Remark}

\newcommand{\bb}{\begin{equation}\begin{aligned}\hspace{0pt}}
\newcommand{\bbb}{\begin{equation*}\begin{aligned}}
\newcommand{\ee}{\end{aligned}\end{equation}}
\newcommand{\eee}{\end{aligned}\end{equation*}}

\newcommand{\ketbra}[1]{\ket{#1}\!\bra{#1}}

\newcommand{\ketbraa}[2]{\ket{#1}\!\bra{#2}}

\newcommand{\sumno}{\sum\nolimits}

\newcommand{\e}{\varepsilon}
\renewcommand{\epsilon}{\varepsilon}
\newcommand{\dd}{\mathrm{d}}

\newcommand{\ve}{\varepsilon}

\DeclareMathOperator{\Tr}{Tr}

\DeclareMathAlphabet{\pazocal}{OMS}{zplm}{m}{n}

\DeclareMathOperator{\supp}{supp}

\newcommand{\lsmatrix}{\left(\begin{smallmatrix}}
\newcommand{\rsmatrix}{\end{smallmatrix}\right)}

\newcommand{\deff}[1]{\textbf{\emph{#1}}}

\stackMath
\newcommand\xxrightarrow[2][]{\mathrel{%
  \setbox2=\hbox{\stackon{\scriptstyle#1}{\scriptstyle#2}}%
  \stackunder[5pt]{%
    \xrightarrow{\makebox[\dimexpr\wd2\relax]{$\scriptstyle#2$}}%
  }{%
   \scriptstyle#1\,%
  }%
}}

\newcommand{\tends}[2]{\xxrightarrow[\! #2 \!]{\mathrm{#1}}}

\stackMath

\makeatletter
\newcommand*\rel@kern[1]{\kern#1\dimexpr\macc@kerna}
\newcommand*\widebar[1]{%
  \begingroup
  \def\mathaccent##1##2{%
    \rel@kern{0.8}%
    \overline{\rel@kern{-0.8}\macc@nucleus\rel@kern{0.2}}%
    \rel@kern{-0.2}%
  }%
  \macc@depth\@ne
  \let\math@bgroup\@empty \let\math@egroup\macc@set@skewchar
  \mathsurround\z@ \frozen@everymath{\mathgroup\macc@group\relax}%
  \macc@set@skewchar\relax
  \let\mathaccentV\macc@nested@a
  \macc@nested@a\relax111{#1}%
  \endgroup
}

\counterwithin*{equation}{part}
\counterwithin*{thm}{part}
\counterwithin*{figure}{part}

\definecolor{Blues5seq1}{RGB}{239,243,255}
\definecolor{Blues5seq2}{RGB}{189,215,231}
\definecolor{Blues5seq3}{RGB}{107,174,214}
\definecolor{Blues5seq4}{RGB}{49,130,189}
\definecolor{Blues5seq5}{RGB}{8,81,156}

\definecolor{Greens5seq1}{RGB}{237,248,233}
\definecolor{Greens5seq2}{RGB}{186,228,179}
\definecolor{Greens5seq3}{RGB}{116,196,118}
\definecolor{Greens5seq4}{RGB}{49,163,84}
\definecolor{Greens5seq5}{RGB}{0,109,44}

\definecolor{Reds5seq1}{RGB}{254,229,217}
\definecolor{Reds5seq2}{RGB}{252,174,145}
\definecolor{Reds5seq3}{RGB}{251,106,74}
\definecolor{Reds5seq4}{RGB}{222,45,38}
\definecolor{Reds5seq5}{RGB}{165,15,21}

\allowdisplaybreaks

\let\nc\newcommand
\renewcommand{\bar}{\;\rule{0pt}{9.5pt}\right|\;}
\nc{\lset}{\left\{\left.}
\nc{\rset}{\right\}}
\nc{\lsetr}{\left\{\,}
\nc{\rsetr}{\right.\right\}}
\nc{\barr}{\;\rule{0pt}{9.5pt}\left|\;}

\newcommand{\wt}{\widetilde}

\let\textleq\relax

\let\texteq\relax

\newcommand{\texteq}[1]{\stackrel{\mathclap{\mbox{\scriptsize#1}}}{=}}
\newcommand{\textleq}[1]{\stackrel{\mathclap{\mbox{\scriptsize#1}}}{\leq}}

\usepackage[most,breakable]{tcolorbox}
\renewenvironment{boxed}[1][white]%
  {\expandafter\ifstrequal\expandafter{#1}{filled}{\begin{tcolorbox}[colback=MidnightBlue!70!black!70!TealBlue!2!white,colframe=MidnightBlue!70!black!70!TealBlue!30!white,breakable=false,enhanced,left=5.75pt,right=5.75pt,grow sidewards by=10pt]}{\begin{tcolorbox}[colback=white,colframe=gray!15,breakable,enhanced,left=5.75pt,right=5.75pt,grow sidewards by=10pt]}}%
  {\end{tcolorbox}}

\nc{\R}{\mathcal{R}}
\nc{\W}{\mathcal{W}}
\nc{\T}{\mathcal{T}}
\nc{\B}{\mathcal{B}}
\nc{\C}{\mathcal{C}}
\nc{\U}{\mathcal{U}}
\nc{\E}{\mathcal{E}}
\nc{\EE}{\mathscr{E}}
\nc{\K}{\mathcal{K}}
\nc{\V}{\mathcal{V}}
\nc{\X}{\mathcal{X}}
\nc{\F}{\mathcal{F}}
\nc{\G}{\mathcal{G}}
\nc{\D}{\mathcal{D}}
\nc{\Y}{\mathcal{Y}}

\nc{\M}{\mathcal{M}}
\nc{\N}{\mathcal{N}}
\nc{\I}{\mathcal{I}}
\nc{\Q}{\mathbb{Q}}
\nc{\RR}{\mathbb{R}}
\nc{\CC}{\mathbb{C}}

\nc{\HH}{\mathbb{H}}
\nc{\MM}{\mathbb{M}}
\nc{\NN}{\mathbb{N}}
\nc{\DD}{\mathbb{D}}

\nc{\id}{\mathbbm{1}}
\nc{\idc}{\mathrm{id}}

\nc{\norm}[2]{\left\lVert#1\right\rVert_{\,#2}}
\nc{\proj}[1]{\ket{#1}\!\bra{#1}}
\nc{\lnorm}[2]{\left\lVert#1\right\rVert_{\ell_{#2}}}

\let\oldproofname\proofname
\renewcommand{\proofname}{\rm\bf{\oldproofname}}

\makeatletter
\renewenvironment{proof}[1][\proofname]{\par
\pushQED{\qed}%
\normalfont \topsep6\p@\@plus6\p@\relax
\trivlist
\item\relax
{\bfseries  
#1\@addpunct{.}}\hspace\labelsep\ignorespaces 
}{%
\popQED\endtrivlist\@endpefalse
}
\makeatother

\nc{\rhos}{\rho'}

\nc{\gm}{\mathbin\#}

\newcommand{\sPnorm}{P,=}
\newcommand{\sPsub}{P,\leq}
\newcommand{\sTnorm}{T,=}
\newcommand{\sTsub}{T,\leq}
\newcommand{\smoothing}[2]{#2,\,#1}
\NewDocumentCommand{\s}{m m O{\ve}}{%
 \lowercase{\def\sdist{#1}}
 \lowercase{\def\snorm{#2}}
 \IfEqCase{\sdist}{%
  {t}{%
   \IfEqCase{\snorm}{%
    {norm}{ \smoothing{\sTnorm}{#3} }%
    {n}{ \smoothing{\sTnorm}{#3} }%
    {sub}{ \smoothing{\sTsub}{#3}}%
    {s}{ \smoothing{\sTsub}{#3}}%
   }%
  }%
  {p}{%
   \IfEqCase{\snorm}{%
    {norm}{ \smoothing{\sPnorm}{#3} }%
    {n}{ \smoothing{\sPnorm}{#3} }%
    {sub}{ \smoothing{\sPsub}{#3}}%
    {s}{ \smoothing{\sPsub}{#3}}%
   }%
  }%
 }%
}%
\NewDocumentCommand{\Dmax}{m m O{\ve}}{%
 D_{\max}^{\s{#1}{#2}[#3]}
}%

\makeatletter

\def\@sect@ltx#1#2#3#4#5#6[#7]#8{%
    \@ifnum{#2>\c@secnumdepth}{%
        \def\H@svsec{\phantomsection}%
        \let\@svsec\@empty
    }{%
        \H@refstepcounter{#1}%
        \def\H@svsec{%
            \phantomsection
        }%
        \protected@edef\@svsec{{#1}}%
        \@ifundefined{@#1cntformat}{%
            \prepdef\@svsec\@seccntformat
        }{%
            \expandafter\prepdef
            \expandafter\@svsec
            \csname @#1cntformat\endcsname
        }%
    }%
    \@tempskipa #5\relax
    \@ifdim{\@tempskipa>\z@}{%
        \begingroup
        \interlinepenalty \@M
        #6{%
            \@ifundefined{@hangfrom@#1}{\@hang@from}{\csname @hangfrom@#1\endcsname}%
            {\hskip#3\relax\H@svsec}{\@svsec}{#8}%
        }%
        \@@par
        \endgroup
        \@ifundefined{#1mark}{\@gobble}{\csname #1mark\endcsname}{#7}%
        \addcontentsline{toc}{#1}{%
            \@ifnum{#2>\c@secnumdepth}{%
                \protect\numberline{}%
            }{%
                \protect\numberline{\csname the#1\endcsname}%
            }%
            #7}
    }{%
        \def\@svsechd{%
            #6{%
                \@ifundefined{@runin@to@#1}{\@runin@to}{\csname @runin@to@#1\endcsname}%
                {\hskip#3\relax\H@svsec}{\@svsec}{#8}%
            }%
            \@ifundefined{#1mark}{\@gobble}{\csname #1mark\endcsname}{#7}%
            \addcontentsline{toc}{#1}{%
                \@ifnum{#2>\c@secnumdepth}{%
                    \protect\numberline{}%
                }{%
                    \protect\numberline{\csname the#1\endcsname}%
                }%
                #8}%
        }%
    }%
    \@xsect{#5}}%

\makeatother